\documentclass[acmtog,nonacm,balance=false]{acmart}
\usepackage{amsmath}
\usepackage{algorithm}
\usepackage{algpseudocode}

\setcopyright{none}
\newcommand{\wopfigure}[2]{\IfFileExists{#1}{\includegraphics[width=\linewidth,height=#2,keepaspectratio]{#1}}{\fbox{\parbox[c][#2][c]{0.94\linewidth}{\centering\small Figure asset not included\\\texttt{\detokenize{#1}}}}}}
\begin{document}
\title{Walk-on-Ambient-Patch: Solver-Adapted Isometric Embeddings for Projected Walk on Spheres}
\author{Yunfeng Liao}
\author{Xiucheng Li}
\affiliation{\institution{Harbin Institute of Technology (Shenzhen)}\city{Shenzhen}\state{Guangdong}\country{China}}
\renewcommand{\shortauthors}{Liao and Li}
\begin{abstract}
Projected Walk on Spheres (PWoS) solves surface PDEs through ambient sampling and closest-point projection, but tightly folded embeddings can force short steps or ambiguous projections despite simple intrinsic geometry. We introduce \emph{Walk-on-Ambient-Patch} (\textbf{WoAP}), which treats the embedding as a solver variable. Edge-isometry constraints and tangent-point repulsion separate intrinsically distant regions, optionally in higher codimension; patch states track residual projection branches without introducing PDE boundaries. We analyze the resulting geometry--sampling tradeoff: additional dimensions provide separation freedom but reduce uniform tangent progress, while strong tangent-biased proposals incur large importance-weight moments. Folded-surface experiments demonstrate reduced projection conflicts, a separate hard-constraint pilot verifies metric preservation, and an internal fixed-RMSE comparison favors higher dimensions at stricter accuracy targets. Together, these results motivate choosing the embedding and projection state jointly with the stochastic solver.
\end{abstract}

\begin{teaserfigure}
    \centering
    \includegraphics[width=0.9\textwidth]{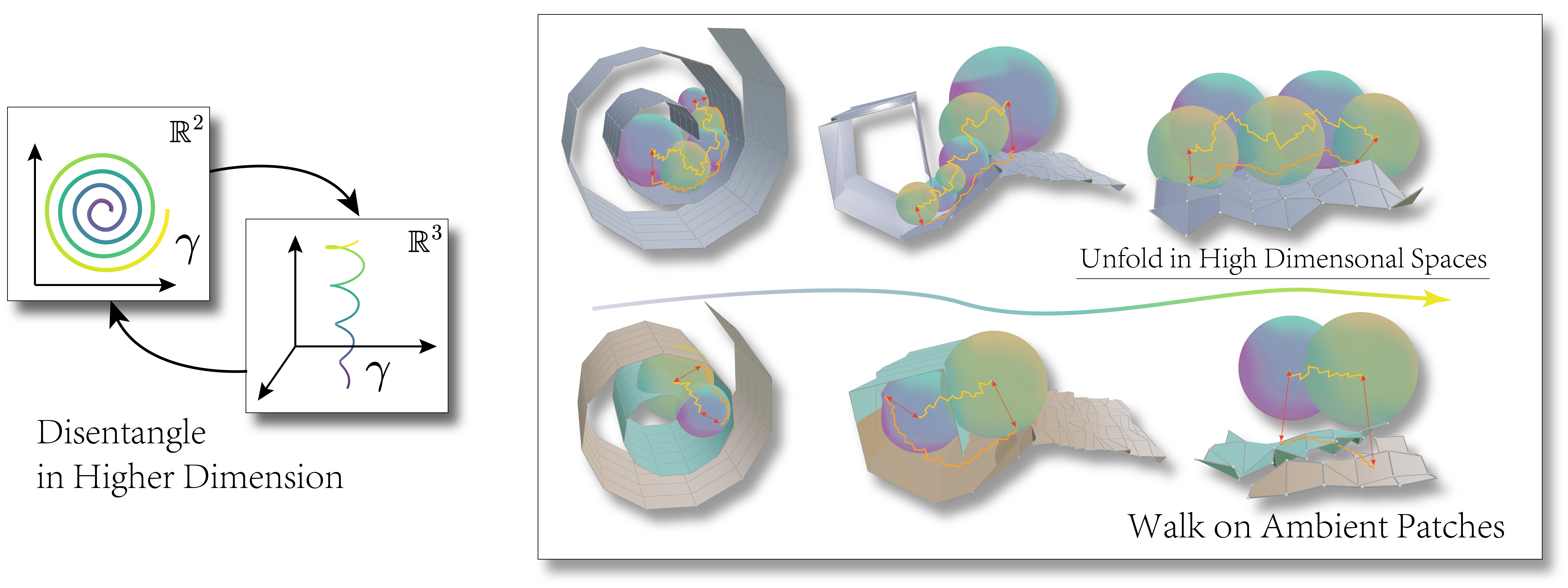}
    \caption{
        \textbf{Walk-on-Ambient-Patch.}
        Higher-dimensional re-embedding separates intrinsically distant regions while preserving the intrinsic metric; local ambient patches handle residual projection ambiguities. Left: higher-dimensional disentanglement. Right: solver-adapted re-embeddings of a Swiss-roll, allowing larger projected walks over local ambient patches as it unfolds.
    }
    \Description{Left: higher-dimensional disentanglement. Right: solver-adapted re-embeddings and projected walks over local ambient patches.}
    \label{fig:teaser}
\end{teaserfigure}

\maketitle
\thispagestyle{fancy}\fancyfoot[C]
{\footnotesize Preprint}

\section{Introduction}
\label{sec:introduction}
Pointwise Monte Carlo solvers make PDE-based geometry processing possible without constructing a volumetric mesh or solving for the entire field~\cite{SawhneyCrane2020MCGP}. Projected Walk on Spheres (PWoS) extends this idea to surfaces: sample an ambient neighborhood, project back to the surface, and recursively estimate the solution~\cite{Sugimoto2024PWoS}. Its geometric queries, however, expose a mismatch between the problem being solved and the representation used to solve it. The PDE depends on the intrinsic metric, whereas the usable projection neighborhood depends on the extrinsic embedding. The uniqueness radius of closest-point projection is formalized by the reach, which couples curvature and proximity to the medial axis~\cite{Federer1959CurvatureMeasures,Gray2004Tubes}. The admissible tubular neighborhood can shrink for two distinct reasons. Locally, curvature bounds the focal radius before normal characteristics intersect; globally, geodesically distant patches can approach each other in the ambient space, so their normal neighborhoods collide and the closest-point map becomes multivalued.

\begin{figure}
    \centering
    \includegraphics[width=0.8\linewidth]{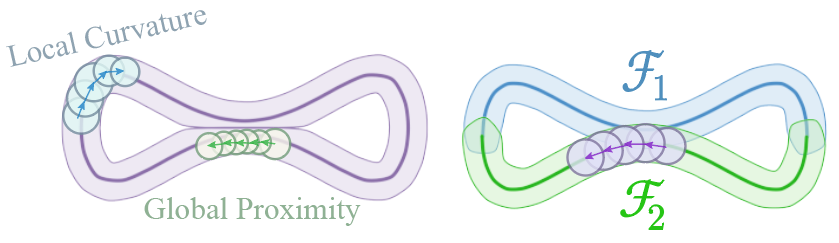}
    \caption{Two bottlenecks of projected walks. Local curvature limits the focal radius, while global self-proximity creates competing closest-point branches. Patch splitting such that $M=\mathcal F_1\cup\mathcal F_2$ alleviates the global bottleneck and enlarges valid walks, but local curvature still requires re-embedding.}
    \label{fig:tubular_neighborhood_patch}
\end{figure}

Consider a sheet folded into closely spaced layers. Diffusion between adjacent layers must follow the surface, even when their ambient separation is tiny. A global closest-point query at a large random endpoint can jump to a remote layer; restricting every step to a globally unambiguous neighborhood instead sacrifices intrinsic progress, as illustrated in Figure~\ref{fig:tubular_neighborhood_patch}. Unfolding or separating the layers can improve these queries without changing the intrinsic PDE. This motivates our central question: \emph{can we choose the embedding for the stochastic solver, rather than require the solver to accept the input embedding?}

Prior work provides the necessary ingredients but addresses different objectives. Closest-point methods and projected walks in higher dimensions solve PDEs on a given embedded manifold~\cite{RuuthMerriman2008ClosestPoint,Hui2026MPWoS}. Classical isometric-embedding results show that higher codimension provides additional degrees of freedom for metric-preserving embeddings~\cite{ChernKuiper1952IsometricEmbedding,Nash1956IsometricEmbedding}. Repulsive geometric optimization controls self-proximity, including under metric-aware deformation~\cite{Yu2021RepulsiveSurfaces,Sassen2024RepulsiveShells}. Geometry-processing deformation methods likewise preserve local shape, rigidity, or metric structure through rigidity, projection, or developability constraints~\cite{SorkineAlexa2007ARAP,BotschSorkine2008VariationalDeformation,BouazizEtAl2012ShapeUp,KilianEtAl2008CurvedFolding,JiangEtAl2020IsometricMappings}. We combine these viewpoints by optimizing extrinsic geometry subject to a fixed discrete intrinsic metric. While re-embedding separates globally competing surface regions and patch states resolve residual projection ambiguities, the additional degrees of freedom afforded by isometric embeddings may also alleviate local curvature constraints on admissible walk radii.

Our representation therefore has two coupled parts. First, edge-constrained tangent-point optimization separates nonlocal surface regions. Second, a phase-field partition assigns remaining competing branches distinct projection states. Most steps use one owner-restricted projection; near transitions, deterministic continuation follows the branch of the same random endpoint. Patch interfaces encode changes of projection branch. The stochastic backend retains the dimension-dependent Green recursion and numerical compensation of Hui et al.~\cite{Hui2026MPWoS}.

Geometric separation alone is not an adequate objective. For intrinsic dimension $m$ and ambient dimension $D$, a uniform direction carries only $m/D$ expected squared tangent component. Biasing directions toward the tangent space increases progress but introduces multiplicative importance weights. Our analysis quantifies this tradeoff in a flat local model and motivates evaluating the representation at fixed solution accuracy. We evaluate the resulting solver on both approximate and exact-isometry representations, and examine accuracy-dependent computational cost.

Our contributions are:
\begin{itemize}
\item \textbf{Solver-adapted embedding.} A formulation with edge constraints uses extrinsic freedom to reduce the projection conflicts that limit projected PDE walks.
\item \textbf{Projection-state representation.} Phase domains guided by conflicts and a fast/transition policy track residual local branches with the same stochastic sample, without adding PDE boundaries.
\item \textbf{Geometry--sampling analysis.} Exact formulas for tangent progress and importance weights, with dimension-scaling results, explain why clearance and walk length must be assessed jointly with sampling variance.
\end{itemize}
Sections~\ref{sec:related_work}--\ref{sec:preliminary} review prior work and define the projected recursion. Section~\ref{sec:embedding} develops the embedding and projection states, and Section~\ref{sec:analysis} analyzes their sampling tradeoff. Section~\ref{sec:experiments} evaluates the components and accuracy-dependent cost; the appendices provide proofs, implementation details, and expanded experiments.

\section{Related Work}
\label{sec:related_work}
\paragraph{Grid-free PDE solvers.}
Classical WoS~\cite{Muller1956WoS} was developed into a grid-free framework for geometry processing~\cite{SawhneyCrane2020MCGP}. Extensions address spatially varying coefficients~\cite{SawhneyEtAl2022VariableCoeff}, mixed Dirichlet--Neumann conditions~\cite{SawhneyEtAl2023WoSt}, and Robin conditions~\cite{MillerEtAl2024Robin}. Boundary-value caching~\cite{MillerEtAl2023BVC} and harmonic caching~\cite{ZhouEtAl2025HarmonicCaching} reuse estimates spatially. Walk on Decomposed Subdomains couples Monte Carlo local solution operators to a deterministic global solve~\cite{JambonEtAl2026WoDS}.WoAP instead changes the geometric representation used by the pointwise solver: its projection labels select closest-point branches of an auxiliary embedding.

\paragraph{Surface PDEs and closest-point extensions.}
Surface finite elements discretize the intrinsic differential operator~\cite{Dziuk1988Beltrami,DziukElliott2013SurfaceFEM}; the Closest Point Method and related embedding formulations extend surface quantities into an ambient neighborhood~\cite{BertalmioEtAl2001ImplicitSurfacePDE,RuuthMerriman2008ClosestPoint,MacdonaldRuuth2009ImplicitCP,MarzMacdonald2012ClosestPointCalculus}. PWoS combines this extension with pointwise stochastic recursion~\cite{Sugimoto2024PWoS}. Hui et al.~\cite{Hui2026MPWoS} develop a high-dimensional projected formulation with off-manifold compensation and dimension-dependent Green weights. We adopt this backend and optimize its geometric representation.

\paragraph{Repulsive geometry and projection-state partitions.}
Tangent-point energies provide a principled way to penalize nonlocal self-proximity in arbitrary codimension~\cite{Strzelecki2013TangentPoint}, while Repulsive Surfaces~\cite{Yu2021RepulsiveSurfaces} and Repulsive Shells~\cite{Sassen2024RepulsiveShells} develop efficient optimization and metric-aware deformation. Related mesh-deformation and developable-surface methods preserve local shape or metric structure under geometric editing~\cite{SorkineAlexa2007ARAP,BotschSorkine2008VariationalDeformation,BouazizEtAl2012ShapeUp,KilianEtAl2008CurvedFolding,JiangEtAl2020IsometricMappings}. Building on these methods, WoAP optimizes embedding clearance under intrinsic metric constraints. The remaining projection branches are encoded by a conflict-aware phase field, using graph regularization inspired by surface segmentation~\cite{Meny2021PhaseField,GarciaCardonaEtAl2014MulticlassGraph,DondlWojtowytsch2021Connectedness}. Its labels control the closest-point queries used by the stochastic solver: owner-restricted evaluation on fast faces and local continuation near transitions.

\section{Preliminaries}
\label{sec:preliminary}
\subsection{Intrinsic Problem and Probabilistic Representation}
Let $(M,g)$ be a compact $m$-manifold with physical boundary $\partial M$, embedded by $X:M\to\mathbb R^{D_0}$. We consider the surface Poisson problem with constant screening $\sigma\ge0$ and Dirichlet data $b$:
\begin{equation}
-\Delta_Mu+\sigma u=f\quad\text{in }M,\qquad u=b\quad\text{on }\partial M.
\label{eq:intrinsic_pde}
\end{equation}
For a diffusion $B_t$ with generator $\Delta_M$ and exit time $\tau$, the Feynman--Kac representation is
\begin{equation}
u(x)=\mathbb E_x\!\left[e^{-\sigma\tau}b(B_\tau)+\int_0^\tau e^{-\sigma t}f(B_t)\,dt\right].
\label{eq:fk}
\end{equation}
The generator convention fixes the time normalization. We identify $M$ with its original embedded realization
$X(M)\subset\mathbb R^{D_0}$ and seek an alternative isometric embedding $\iota:M\to\mathbb R^D$. Thus $\iota(x)$ maps each point $x$ from its original ambient position to a new realization, while $f$ and $b$ remain attached to the same material coordinates. 

Our implementation uses a triangle mesh $(V,E,F)$, where edge-length preservation retains the piecewise-flat intrinsic metric. More generally, the method only requires PWoS-compatible closest-point queries and local branch information. Explicit representations, including meshes, oriented point clouds, and surface-aligned 3D Gaussians, additionally need a neighborhood graph for intrinsic connectivity and a compact patch-ownership representation. Implicit surfaces instead require local projection/tangent/normal queries and an implicit phase field returning the mixture of candidate projection states at each spatial query.

\subsection{WoS and the Projected Green Recursion}
For a harmonic function in a Euclidean ball, the center value equals the uniform boundary average. WoS recursively samples this average, choosing balls inside the physical domain~\cite{Muller1956WoS}. PWoS instead extends the surface solution through a local closest-point map $\pi$~\cite{Sugimoto2024PWoS} where the admissible walking step at $x$ is limited by both the ambient distance tubular neighborhood boundary and the geodesic distance $\hat d_{\partial M}(x)$ to the boundary $\partial M$. Write $Eu=u\circ\pi$ and $Ef=f\circ\pi$. On a smooth tubular neighborhood, $\Delta_{\mathbb R^D}Eu=\Delta_Mu$ only on the surface, and the off-manifold tangential residual is
\begin{equation}
 \mathfrak m=-\Delta_{\mathbb R^D}Eu+\sigma Eu-Ef,\qquad \mathfrak m|_M=0.
\label{eq:compensation}
\end{equation}
The $D$-dimensional Green identity used by Hui et al.~\cite{Hui2026MPWoS} gives, at $y=\iota(x)$,
\begin{equation}
\begin{split}
u(x)={}&c_{\sigma,r}\mathbb E_{\omega\sim p_D}u(\pi(y+r\omega))+a_{\sigma,r}\mathbb E_{Z\sim\rho_{D,\sigma,r}}[f(\pi(Z))+\mathfrak m(Z)],
\end{split}
\label{eq:hui_local_recursion}
\end{equation}
where $p_D$ is uniform on $\mathbb S^{D-1}$, $\rho=G_{D,\sigma,r}/a_{\sigma,r}$ is the normalized positive Dirichlet Green density, and $ c_{\sigma,r}={(\sqrt\sigma r)^\nu}/({2^\nu\Gamma(\nu+1)I_\nu(\sqrt\sigma r)}),\allowbreak a_{\sigma,r}={(1-c_{\sigma,r})}/{\sigma},\allowbreak \nu=D/2-1.$ At $\sigma=0$, these limits are $c_{0,r}=1$ and $a_{0,r}=r^2/(2D)$. Thus changing $D$ changes both boundary and source weights. In practice, $\mathfrak m$ is replaced by a normal Taylor approximation $\mathfrak m^*$ obtained from an independent coarse predictor, following Hui et al.; Appendix~\ref{app:backend} specifies our realization. % In a nutshell, PWoS requires each integration ball to lie within a locally valid closest-point tubular neighborhood; our method enlarges and disambiguates these valid neighborhoods.

\section{Walk on Ambient Patch via Solver-Adapted Embeddings}
\label{sec:embedding}\label{sec:solver_embedding}\label{sec:wop}
Our goal is to enlarge the ambient steps available to PWoS while preserving the intrinsic PDE and avoiding projection jumps between nearby surface sheets. We first optimize an auxiliary embedding with edge-length constraints to separate intrinsically distant regions. Residual projection conflicts are then resolved by phase-field projection domains, with an offline fast/transition map determining where owner-restricted projection suffices and where local branch tracking is required. At runtime, we carry the surface point, seed face, and projection owner, and continue each sampled ambient step across patch interfaces without resampling. Intrinsic coefficients, sources, and boundary data remain attached to the original mesh through unchanged connectivity and barycentric material coordinates.

\subsection{Edge-Constrained Repulsive Embedding}
\label{sec:general_tp}

Our first objective is to enlarge the projection neighborhoods available to
PWoS without changing the intrinsic PDE. Since the intrinsic PDE is determined by the surface metric, we keep every input edge length $\ell_{ij}=\|x_i-x_j\|$ fixed and optimize only the extrinsic realization $Y=(y_i)_{i\in V}\in\mathbb R^{|V|\times D}$, where $y_i=\iota(x_i)$. The desired embedding should separate surface
regions that are intrinsically unrelated but become close in the ambient
space. Directly maximizing the minimum valid closest-point radius is
nonsmooth and depends on global projection events, so we instead optimize a
smooth repulsive surrogate under exact metric constraints, i.e., preserving discrete isometry:
\begin{equation}
\min_Y \mathcal E_{\rm rep}(Y)
\quad\text{s.t.}\quad
c_{ij}(Y):=\|y_i-y_j\|^2-\ell_{ij}^2=0,
\quad (i,j)\in E .
\label{eq:exact_isometry}
\end{equation}

\paragraph{Projection-clearance surrogate.}
For each face $f$, let $q_f$, $A_f$, and $P_T^f$ denote its barycenter, area,
and tangent projector, and let $P_N^f=I-P_T^f$. For a nonlocal face pair
$(f,g)$ with $d_{fg}=q_g-q_f$, we measure how strongly $g$ approaches the
normal bundle of $f$ by the tangent-point radius
\begin{equation}
R_{tp}(f,g)
=
\frac{\|d_{fg}\|^2}
     {2\|P_N^f d_{fg}\|}.
\label{eq:discrete_tp_radius}
\end{equation}
Small $R_{tp}(f,g)$ indicates precisely the configuration undesirable for PWoS:
two distinct surface regions approach each other in the ambient space along a
direction nearly normal to one of them. We therefore use the symmetric
discrete energy
\begin{equation}
\mathcal E_{\rm rep}(Y)
=
\frac12
\sum_{(f,g)\in\mathcal C(Y)}
A_fA_g
\left(
R_{tp}(f,g)^{-p}+R_{tp}(g,f)^{-p}
\right), 
\label{eq:discrete_tp_energy}
\end{equation} % p=6 ,
adapted from tangent-point repulsion
\cite{Strzelecki2013TangentPoint,Yu2021RepulsiveSurfaces}.
Here $\mathcal C(Y)$ contains nonlocal face pairs whose current ambient
separation is below a user-defined  interaction range $2R^*$; topologically local
pairs are excluded because they represent the intended intrinsic
neighborhood rather than a projection conflict. We rebuild this candidate set
during optimization using a spatial broad phase. Thus the repulsion is spent
only on pairs capable of limiting the target projection scale $R^*$. The resulting surrogate promotes larger usable projection neighborhoods while preserving intrinsic distances.

\paragraph{Metric-constrained optimization.}
Let $J=\partial c/\partial Y$ be the Jacobian of all edge constraints and
$h=\nabla_Y\mathcal E_{\rm rep}$. We first remove the component of the
repulsive gradient that changes edge lengths, $h_{\parallel}=h-J^\top(JJ^\top)^\dagger Jh ,$
take a descent step $Y\leftarrow Y-\eta h_{\parallel}$, and then retract back
to the isometric constraint manifold by repeated Newton projection, $Y\leftarrow Y-J^\top(JJ^\top)^\dagger c(Y),$  until the maximum edge residual reaches the prescribed tolerance. This is an
equality-constrained projected-gradient scheme
\cite{NocedalWright2006NumericalOptimization}; the contribution here is the
solver-specific objective and its coupling to the fixed intrinsic metric.

For $D>D_0$, the zero-padded input
$y_i=(x_i,0,\ldots,0)$ is already feasible but is symmetric with respect to
the added coordinates. We therefore initialize the new coordinates with a
small perturbation and immediately project back to
$c(Y)=0$, allowing subsequent descent to exploit the additional ambient
degrees of freedom without changing edge lengths. In our hard-constraint
experiments, this procedure preserves the piecewise-flat metric to numerical
tolerance.

\paragraph{Scalable approximate projection.}
For the large-scale sweep, we replace the global Newton retraction by independent edge-wise projections. Although Eq.~\eqref{eq:exact_isometry} uses the squared constraint, it has the same feasible set as $C_{ij}(Y):=\|y_i-y_j\|-\ell_{ij}=0$. Let $d_{ij}=\|y_i-y_j\|$, $n_{ij}=(y_i-y_j)/d_{ij}$. For one edge, we compute the minimum endpoint displacement that exactly restores its target length:
\begin{equation}
\begin{aligned}
\min_{\delta_i,\delta_j}&
\frac12\|\delta_i\|^2+\frac12\|\delta_j\|^2
\quad\text{s.t.}\quad
\|(y_i+\delta_i)-(y_j+\delta_j)\|=\ell_{ij},
\end{aligned}
\label{eq:pbd_local_projection}
\end{equation}
where $\delta_i=-\alpha n_{ij}$ is the correction.
By symmetry, the minimum-displacement solution preserves the edge midpoint, so $\delta_j=-\delta_i$. Substitution into the constraint gives $\|(d_{ij}-2\alpha)n_{ij}\|=\ell_{ij}$ and hence $\alpha=(d_{ij}-\ell_{ij})/2$, $\delta_i=\frac12(d_{ij}-\ell_{ij})(y_j-y_i)/d_{ij}$. We regularize the denominator in implementation and use
\begin{equation}
\Delta_{ij}
=\frac12\frac{d_{ij}-\ell_{ij}}     {d_{ij}+\epsilon}(y_j-y_i),
\qquad
\Delta_{ji}=-\Delta_{ij}.
\label{eq:pbd_edge_projection}
\end{equation}
After each repulsive descent step, we iteratively apply the edge projections in Eq.~\eqref{eq:pbd_edge_projection} using Jacobi sweeps. Because multiple edge constraints share vertices, the full edge-length system is only approximately restored.

% \subsection{Projection Domains and Transition Classification}
% \label{sec:phase_domains}
\subsection{Conflict-Aware Phase-Field Partitioning}

Re-embedding removes most nonlocal projection conflicts, but residual competing closest-point branches may remain. We therefore prescribe $K$ \emph{projection owners} and partition the surface into $K$ locally coherent owner regions. An owner identifies which local closest-point branch should be followed; owner interfaces are solver states rather than physical boundaries of the PDE.

We first construct a residual conflict graph $\mathcal C$ from surface regions that are ambiently close but sufficiently separated in graph-geodesic distance. With $R^*$ denoting the target projection-radius cap, $w_{\max}$ capping excessively large conflict weights, and $p_c>0$ the clearance sensitivity exponent, for each candidate pair $(i,j)$,  we define
\begin{align}
    R_{\rm sym}(i,j)& =\min\{R_{\rm tp}(i,j),R_{\rm tp}(j,i)\}, \\
w_{ij}& =\min\{w_{\text{max}},(R^*/R_{\rm sym}(i,j))^{p_c}\},
\label{eq:phase_conflict_weight}
\end{align}
so pairs with smaller projection clearance receive stronger separation penalties. Each vertex carries a soft owner assignment $\Phi_i=(\phi_i^1,\ldots,\phi_i^K)$ with $\phi_i^a\ge0$ and $\sum_a\phi_i^a=1$. With $\widetilde w$ normalized to unit mean, our solver-specific phase objective is
\begin{equation}
E_{\rm conf}(\Phi)=\frac1{|\mathcal C|}\sum_{(i,j)\in\mathcal C}\widetilde w_{ij}\,\Phi_i^\top\Phi_j,
\label{eq:phase_conflict_energy}
\end{equation}
which directly penalizes assigning competing projection branches to the same owner. Conflict separation alone can produce fragmented or diffuse owner fields, so we optimize
\begin{equation}
E_{\rm phase}(\Phi)=\lambda_{\rm conf}E_{\rm conf}(\Phi)+\lambda_{\rm reg}\mathcal R(\Phi),
\label{eq:phase_total_energy}
\end{equation}
where $\mathcal R$ is a fixed multiphase regularizer and its internal implementation weights are absorbed into its definition. Specifically, $\mathcal R$ combines graph smoothness $\sum_{(i,j)\in E}\|\Phi_i-\Phi_j\|_2^2/\ell_{ij}$, a multiwell term $\sum_{i,a}(\phi_i^a)^2(1-\phi_i^a)^2$, area balance $\sum_a(m_a-A/K)^2$ with $A=\sum_iA_i$ and $m_a=\sum_iA_i\phi_i^a$, and seed anchoring $-\sum_a\log(\phi_{s_a}^a+\epsilon)$. We further encourage compact and connected owner regions using $\sum_{i,a}A_i\phi_i^a\bar d_a(i)^2$ and $\sum_{i,a}A_i\phi_i^a\bar d_a^{\rm conn}(i)$, where $s_a$ is the geodesic farthest-point seed of owner $a$, $\bar d_a$ is its normalized graph-geodesic distance, and $\bar d_a^{\rm conn}$ is a phase-dependent shortest-path distance that penalizes paths through low-membership vertices. Following phase-field connectedness regularization~\cite{DondlWojtowytsch2021Connectedness}, these distances are periodically recomputed with Dijkstra~\cite{Dijkstra1959ShortestPath} and frozen between optimization updates. % Softmax logits enforce the simplex constraint.

\paragraph{Initialization and optimization.}
We initialize the $K$ owner fields from geodesic farthest-point seeds $s_1,\ldots,s_K$. Let $d_G(i,s_a)$ denote the graph-geodesic distance from vertex $i$ to seed $s_a$; the initial logits are chosen as a decreasing function of $d_G(i,s_a)$ and normalized by softmax, yielding a soft geodesic Voronoi partition. We optimize these logits directly and set $\Phi_i=\operatorname{softmax}(z_i)$ throughout optimization, which automatically enforces $\phi_i^a\ge0$ and $\sum_a\phi_i^a=1$. % Phase-dependent shortest-path distances in the connectedness regularizer are periodically recomputed and frozen between updates.

After optimization, each face receives the hard owner
\begin{equation}
o_f
=
\arg\max_{a\in\{1,\ldots,K\}}
\sum_{i\in f}\phi_i^a .
\label{eq:hard_owner}
\end{equation}
For owner $a$, let
$\mathcal F_a=\{f\in F:o_f=a\}$ denote its owned face set. The corresponding
owner-restricted closest-point projector is
\begin{equation}
\Pi_a(q)
=
(x^*,\tau^*),
\qquad
(x^*,\tau^*)
\in
\arg\min_{\substack{x\in\tau, \tau\in\mathcal F_a}}
\|\iota(x)-q\|^2 ,
\label{eq:owner_projector}
\end{equation}
where $x^*\in M$ is the projected surface point in material coordinates, and
$\tau^*\in F$ is its landing face. Thus $\Pi_a$ is an ordinary closest-point
query restricted to the primitives assigned to owner $a$. The phase field therefore determines the admissible projection branch for each query.

\subsection{Transition Map Construction}
A hard owner is sufficient only when the owner-restricted projection remains on the same branch throughout a full admissible step. To identify faces where a step may cross into another owner, we trace projected paths along a fixed set of ambient
directions $\Omega=\{\omega_1,\ldots,\omega_{N_\Omega}\}\subset\mathbb S^{D-1}$. Since such transitions primarily occur near owner interfaces, termed \textit{owner seams}, which is an edge shared by two faces with different hard owners.  We test faces within an \textit{H}-ring neighborhood of these seams, where branch changes are most likely to occur. 

Let $F_{\rm band}\subset F$ contain all faces within $H$ face-adjacency hops of these seams. For each face $f\in F_{\rm band}$, let
$x_f\in f$ denote its barycenter in material coordinates, so $\iota(x_f)\in\mathbb R^D$ is its ambient position under the optimized
embedding. For each$f\in F_{\rm band}$ and $\omega\in\Omega$, we consider the ambient segment with $t\in[0,1]$:
\begin{equation}
q_{f,\omega}(t)
=
\iota(x_f)+t\cdot r_f\cdot \omega,\quad  r_f
=
\min\{
R^*,
(1-\delta)\widehat d_{\partial M}(x_f)\}, 
\label{eq:transition_test_segment}
\end{equation}
This segment is sampled at the deterministic parameters
$t_\ell=\ell/N_t$, $\ell=1,\ldots,N_t$. To follow one projection branch along this segment, we use a seeded local
closest-point operator
\begin{equation}
\begin{split}
\operatorname{CP}_{\rm local}(q;\tau)
=(x^*,\tau^*) = \arg\min_{\substack{x \in  f\in\mathcal N_h(\tau)}}
\|\iota(x)-q\|^2.
\end{split}
\label{eq:local_cp}
\end{equation}
where $q\in\mathbb R^D$ is an ambient query and $\tau\in F$ is a seed face.
The solve searches only the local face neighborhood connected to $\tau$ and
returns both the projected material point $x^*\in M$ and its landing face
$\tau^*\in F$. With $\mathcal N_h(\tau)$ denoting all faces within $h$ adjacency rings of
$\tau$, starting from $\tau_0=f$, we recursively evaluate
\begin{equation}
(x_\ell,\tau_\ell)
=
\operatorname{CP}_{\rm local}
\!\left(q_{f,\omega}(t_\ell);\tau_{\ell-1}\right),
\qquad
\ell=1,\ldots,N_t .
\label{eq:offline_transition_track}
\end{equation}
Each solve is therefore seeded by the previous landing face. This is what we
mean by \emph{local continuation}: rather than independently recomputing a
global closest point at every $q_{f,\omega}(t_\ell)$, we advance the query
along the ambient segment while continuously following the locally connected
projection branch.

% Mark $f$ as a raw transition face if there exist $\omega\in\Omega$ and $\ell\in\{1,\ldots,N_t\}$ such that $o_{\tau_\ell}\neq o_f .$ Otherwise $f$ is a raw fast candidate. Owner-seam faces are included in the raw transition set by construction. Hence $\widetilde F_{\rm trans}$ consists of tested faces for which at least one discrete segment of ambient length $r_f$ can continuously reach a face owned by another projection branch.

Let $F_{\rm cross}$ be the tested faces whose tracked paths
reach a different owner. We include the owner-seam faces and set
$\widetilde F_{\rm trans}=F_{\rm cross}\cup F_{\rm seam}$.
The final transition map is
$F_{\rm trans}=\operatorname{Dilate}_1(\widetilde F_{\rm trans})$
and $F_{\rm fast}=F\setminus F_{\rm trans}$.

% The construction is deterministic because the candidate faces $F_{\rm band}$, direction set $\Omega$, continuation parameters $t_\ell$, initial seed face $f$, and local closest-point rule are all fixed during preprocessing. No Monte Carlo variable is sampled in constructing $\widetilde F_{\rm trans}$.

We construct the transition map using finite directional tests and seeded local continuation. On the 2,610-face Swiss roll, four-ring seam-band testing reproduces the full-face finite-direction classification while reducing local closest-point queries from 157,456 to 38,608 (Table~\ref{tab:transition_preprocess_cost}).

% Though only finitely many directions and finitely many points along each segment are tested, this map is accurate enough for WoAP as an empirical classifier. Besides, testing only the seam neighborhood reduces preprocessing cost, but may miss transitions originating farther away. However, as in Table~\ref{tab:transition_preprocess_cost}, both produce identical transition maps on the tested geometry, while seam-band filtering substantially reduces local closest-point evaluations.

\begin{table}[t]
\centering
\caption{Cost of offline fast/transition preprocessing on the 2,610-face Swiss roll. We compare exhaustive finite-direction continuation on every face with testing only the four-ring owner-seam band; both produce the same transition map, while the filtered test uses fewer local closest-point solves.}
\label{tab:transition_preprocess_cost}
\small
\begin{tabular}{lrrrr}
\toprule
Strategy & Tested faces & Local CP calls & Time [s] & Agreement \\
\midrule
All primitives & 2610 & 157456 & 1.036 & 100\% \\
4-ring seam band & 753 & 38608 & 0.265 & 100\% \\
\bottomrule
\end{tabular}
\end{table}

\begin{figure}[t]
\centering
\includegraphics[width=0.9\linewidth]{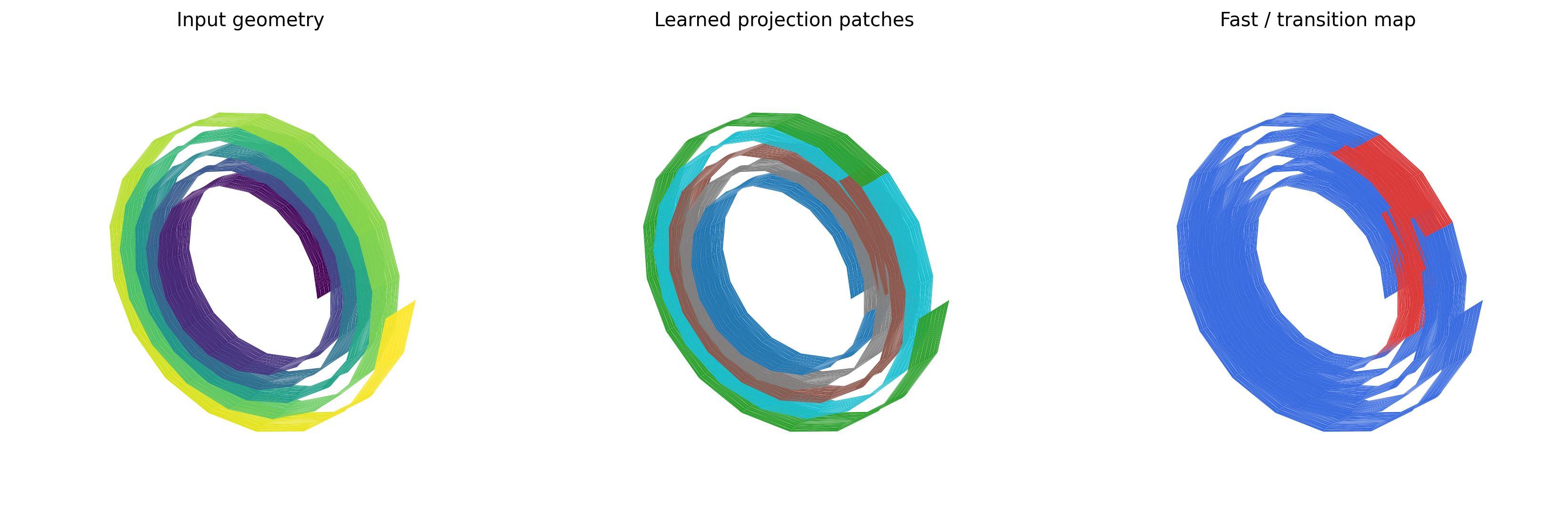}
\caption{Projection-owner partition and offline fast/transition map on a Swiss roll with five patches. }
\label{fig:exp-swiss-geometry}
\Description{Swiss-roll geometry, learned projection patches, and offline fast/transition decomposition.}
\end{figure}

As summarized in Algorithm~\ref{alg:preprocess}, the offline representation therefore consists of the optimized embedding $\iota$, the soft owner field $\Phi$, the hard face labels $\{o_f\}_{f\in F}$, the owner-restricted projectors $\{\Pi_a\}_{a=1}^K$, and the final fast/transition map. An exemplary preprocessed Swiss roll is illustrated in Figure~\ref{fig:exp-swiss-geometry}.

\begin{algorithm}[t]
\caption{Offline preprocessing for WoAP}
\label{alg:preprocess}
\begin{algorithmic}[1]
\Require Mesh $(V,E,F)$, initial embedding
$\iota_0$, target dimension $D$,
owner count $K$, cap $R^*$, test directions $\Omega$,
seam-band width $H$, continuation count $N_t$,
phase iterations $T_{\rm phase}$
\Ensure Optimized embedding $\iota$, phase field $\Phi$,
face owners $\{o_f\}$, projectors $\{\Pi_a\}_{a=1}^K$,
$F_{\rm fast}$, $F_{\rm trans}$,
$\mathfrak m^*$, $\widehat d_{\partial M}$

\State Optimize $\iota:M\rightarrow\mathbb R^D$ from $\iota_0$
using the repulsive embedding solver
\State Build the residual projection-conflict graph $\mathcal C$
\State Choose $K$ geodesic farthest-point seeds and initialize $\Phi$

\For{$t=0,\ldots,T_{\rm phase}-1$}
    \State Periodically refresh the phase-dependent shortest-path distances
    \State Descend on Eq.~\eqref{eq:phase_total_energy}
\EndFor

\State Assign
$o_f\gets\arg\max_a\sum_{i\in f}\phi_i^a$ for every $f\in F$
\State Build $\mathcal F_a=\{f:o_f=a\}$ and $\Pi_a$ for
$a=1,\ldots,K$
\State Build the physical-boundary distance estimate
$\widehat d_{\partial M}$
\State Build the owner-seam $H$-ring candidate set $F_{\rm band}$
\State Initialize $\widetilde F_{\rm trans}$ with all owner-seam faces

\ForAll{$f\in F_{\rm band}$}
    \State
    $r_f\gets
    \min\{R^*,(1-\delta)\widehat d_{\partial M}(x_f)\}$
    \ForAll{$\omega\in\Omega$}
        \State $\tau\gets f$
        \For{$\ell=1,\ldots,N_t$}
            \State
            $q_\ell\gets
            \iota(x_f)+(\ell/N_t)r_f\omega$
            \State
            $(x,\tau)\gets
            \operatorname{CP}_{\rm local}(q_\ell;\tau)$
            \If{$o_\tau\neq o_f$}
                \State Add $f$ to $\widetilde F_{\rm trans}$
                \State \textbf{break}
            \EndIf
        \EndFor
        \If{$f\in\widetilde F_{\rm trans}$}
            \State \textbf{break}
        \EndIf
    \EndFor
\EndFor

\State
$F_{\rm trans}\gets
\operatorname{Dilate}_1(\widetilde F_{\rm trans})$
\State
$F_{\rm fast}\gets F\setminus F_{\rm trans}$

\State Build the coarse predictor and numerical compensation $\mathfrak m^*$

\State \Return
$(\iota,\Phi,\{o_f\},\{\Pi_a\},
F_{\rm fast},F_{\rm trans},
\mathfrak m^*,\widehat d_{\partial M})$
\end{algorithmic}
\end{algorithm}
\subsection{Same-Sample Projection Continuation}
\label{sec:branch}

An owner-restricted projector cannot cross a patch interface: its landing face necessarily retains the current owner. Transition states therefore require a different search rule that permits owner changes through mesh adjacency while avoiding unrestricted jumps between nearby sheets. Algorithm~\ref{alg:walk_on_patch} summarizes the resulting runtime procedure.

%  irrespective of their owners. We define the seeded local projection
% by Eq.~\ref{eq:local_cp}. Unlike $\Pi_a$, this query does not restrict candidate faces to owner $a$.
% It can therefore cross an owner seam through the intrinsic mesh
% neighborhood, while excluding intrinsically remote faces outside that
% neighborhood. The neighborhood size and continuation resolution are
% implementation parameters specified in the appendix.

The runtime state is $(x_k,\tau_k,a_k)$, with $a_k=o_{\tau_k}$ derived
from the containing face. Given the direction-independent radius $r_k=\min\{R^*,(1-\delta)\widehat d_{\partial M}(x_k)\},$
and one sampled direction $\omega_k$, define
\begin{equation}
q_k(t)=\iota(x_k)+t\cdot r_k\omega_k,\qquad t\in[0,1].
\label{eq:runtime_transition_segment}
\end{equation}
On a fast face, one owner-restricted query evaluates the endpoint:
\begin{equation}
(x_{k+1},\tau_{k+1})=\Pi_{a_k}(q_k(1)),
\qquad a_{k+1}=o_{\tau_{k+1}}=a_k.
\label{eq:fast_projection}
\end{equation}
Thus a fast step retains its owner, although its landing face may enter
the transition set.

On a transition face, we follow the same endpoint through deterministic
substeps. Initialize $(\tau^{(0)},a^{(0)})=(\tau_k,a_k)$ and compute
\begin{equation}
\begin{split}
(x^{(\ell)},\tau^{(\ell)})
&=\operatorname{CP}_{\rm local}
\left(q_k(\ell/N_t);\tau^{(\ell-1)}\right),\\
a^{(\ell)}&=o_{\tau^{(\ell)}},
\qquad \ell=1,\ldots,N_t.
\end{split}
\label{eq:runtime_projection_continuation}
\end{equation}
A patch switch occurs precisely when
$a^{(\ell)}\ne a^{(\ell-1)}$. The next local query is seeded by
$\tau^{(\ell)}$, so its candidate neighborhood moves onto the newly
reached patch rather than remaining restricted to the initial owner.
The final state is
\begin{equation}
(x_{k+1},\tau_{k+1},a_{k+1})
=(x^{(N_t)},\tau^{(N_t)},a^{(N_t)}).
\end{equation}

\begin{figure}
    \centering
    \includegraphics[width=0.5\linewidth]{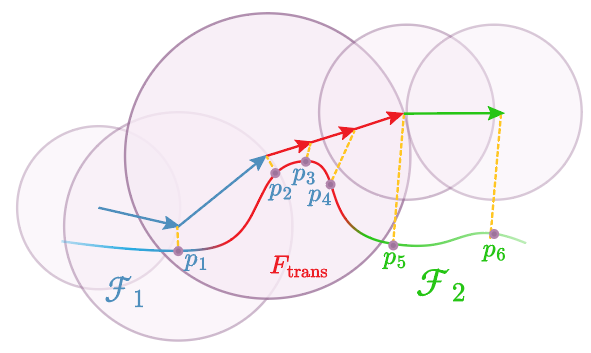}
    \caption{Same-sample continuation across projection owners. Tracking starts at $p_2$ upon entering a transition region; the landing face changes owner at $p_5$.}
    \label{fig:same_sample_proj}
\end{figure}

The owner is updated after each continuation substep, whereas the fast/transition mode is selected only once per stochastic step. Each substep updates the local search neighborhood around the previous landing face, allowing the projection to follow a branch that may lie outside the initial neighborhood. Once tracking begins, it continues to the sampled endpoint even if an intermediate face is fast, since switching to owner-restricted projection may select a different branch, as illustrated in Figure~\ref{fig:same_sample_proj}.

% The endpoint is sampled once; patch switches introduce neither a new
% direction nor an additional importance factor. Finite-step local
% continuation is nevertheless a numerical branch-tracking rule, not a
% certificate of continuity or a globally valid closest-point extension.

The projection rules above supply the next walk state. To evaluate the PDE solution, we combine them with the $D$-dimensional Green recursion of Section~\ref{sec:preliminary}. An independent interior sample $Z_k\sim\rho_{D,\sigma,r_k}$ is locally projected to $z_k=\pi_{\tau_k}(Z_k)$ for source evaluation. The source accumulator $S_k$ and path weight $W_k$ then satisfy 
\begin{equation}
\begin{split}
S_{k+1}
&=S_k+W_k a_{\sigma,r_k}
[f(z_k)+\mathfrak m^*(Z_k)],\\
W_{k+1}
&=W_k c_{\sigma,r_k}w_k,
\qquad \widehat u=S_T+W_Tb(x_T).
\end{split}
\label{eq:hui_wop_source_update}
\end{equation}

\begin{algorithm}[t]
\caption{WoAP runtime and patch switching}
\label{alg:walk_on_patch}
\begin{algorithmic}[1]
\Require Initial $(x,\tau)$, embedding $\iota$, owners $\{o_f\}$,
projectors $\{\Pi_a\}$, $F_{\rm fast}$, $F_{\rm trans}$,
$R^*$, $N_t$, $\widehat d_{\partial M}$, $\mathfrak m^*$
\Ensure One Monte Carlo path contribution $\widehat u$
\State $a\gets o_\tau$; $W\gets1$; $S\gets0$
\While{the physical boundary stopping criterion is not met}
    \State $r\gets\min\{R^*,(1-\delta)\widehat d_{\partial M}(x)\}$
    \State Draw $Z\sim\rho_{D,\sigma,r}$; $z\gets\pi_\tau(Z)$
    \State $S\gets S+Wa_{\sigma,r}[f(z)+\mathfrak m^*(Z)]$
    \State Draw one angular sample $\omega$ and weight $w$
    \State $q\gets\iota(x)+r\omega$
    \If{$\tau\in F_{\rm fast}$}
        \State $(x^+,\tau^+)\gets\Pi_a(q)$
        \State $a^+\gets o_{\tau^+}$
    \Else
        \State $(x^+,\tau^+,a^+)\gets
        \Call{TrackProjection}{\iota(x),q,\tau,N_t}$
    \EndIf
    \State $W\gets Wc_{\sigma,r}w$
    \State $(x,\tau,a)\gets(x^+,\tau^+,a^+)$
    \Comment{next step selects its mode from $\tau$}
\EndWhile
\State \Return $\widehat u=S+Wb(x)$

\Function{TrackProjection}{$q_0,q_1,\tau_0,N_t$}
    \State $\tau\gets\tau_0$; $a\gets o_\tau$
    \For{$\ell=1,\ldots,N_t$}
        \State $q_\ell\gets(1-\ell/N_t)q_0+(\ell/N_t)q_1$
        \State $\mathcal U\gets\mathcal N_h(\tau)$
        \Comment{includes faces of other owners}
        \State $(x^+,\tau^+)\gets
        \arg\min_{f\in\mathcal U,\;x\in f}\|\iota(x)-q_\ell\|^2$
        \State $a^+\gets o_{\tau^+}$
        \If{$a^+\ne a$}
            \State $a\gets a^+$
            \Comment{switch to the landing-face patch}
        \EndIf
        \State $(x,\tau)\gets(x^+,\tau^+)$
        \Comment{move the next search neighborhood}
    \EndFor
    \State \Return $(x,\tau,a)$
\EndFunction
\end{algorithmic}
\end{algorithm}

\section{Geometry--Sampling Tradeoff}
\label{sec:analysis}\label{sec:codim_order}
\subsection{Tangent Progress and Importance Weighting}
\label{sec:vmf_method}
Increasing $D$ gives the embedding more degrees of freedom to separate nearby
layers, but it also changes the geometry of the random walk. The PDE evolves
only along the $m$-dimensional surface, whereas a uniform ambient direction in
$\mathbb R^D$ increasingly points into the normal bundle as $D$ grows; that
component is discarded after projection and produces no intrinsic progress.
Thus higher codimension is useful only when the resulting gain in projection
clearance compensates for this loss of tangent motion.
Theorem~\ref{thm:uniform} quantifies this tradeoff and provides the baseline
against which tangent-biased sampling is evaluated.

\begin{theorem}[Uniform tangent-energy law]
\label{thm:uniform}
Let $T$ be an $m$-dimensional linear subspace of $\mathbb R^D$, $D>m$, and
$\omega\sim\mathrm{Unif}(\mathbb S^{D-1})$. Then $S:=\|P_T\omega\|^2
\sim \mathrm{Beta}\!\left(\frac m2,\frac{D-m}{2}\right)$  and $\mathbb E[S]=\frac mD .$
Hence, for iid flat projected increments
$X_{k+1}-X_k=rP_T\omega_k$, $\mathbb E\|X_N-X_0\|^2=Nr^2\frac mD .$
\end{theorem}

\begin{proof}
By rotational invariance, assume
$T=\operatorname{span}\{e_1,\ldots,e_m\}$ and write
$\omega=Z/\|Z\|$, where $Z_i\stackrel{\mathrm{iid}}{\sim}\mathcal N(0,1)$.
Then $S=\frac{\sum_{i=1}^{m}Z_i^2}
{\sum_{i=1}^{D}Z_i^2}
=
\frac{U}{U+V}$ with $U\sim\chi_m^2,V\sim\chi_{D-m}^2,$ with $U$ and $V$ independent. The standard ratio characterization gives
$S\sim\mathrm{Beta}(m/2,(D-m)/2)$ and $\mathbb E[S]=m/D$. $\mathbb E[P_T\omega_k]=0$ by antipodal symmetry, so all cross terms vanish:
$\mathbb E\|X_N-X_0\|^2=r^2\mathbb E\left\|\sum_{k=1}^{N}P_T\omega_k\right\|^2
=Nr^2\,\mathbb E\|P_T\omega\|^2=Nr^2\frac mD.$
\end{proof}

\begin{theorem}[Curvature-corrected projected-step law]
\label{thm:curved-step}
Let $\iota:M^m\to\mathbb R^D$ be a smooth isometric embedding and let $\pi$ be a
smooth closest-point branch around $\iota(x)$. For $\omega\sim\mathrm{Unif}(\mathbb S^{D-1})$, define
$\xi_r(\omega)=\exp_x^{-1}(\pi(\iota(x)+r\omega))\in T_xM$.
For sufficiently small $r$, 
\begin{equation}
\mathbb E[\xi_r\otimes\xi_r]=\frac{r^2}{D}I_{T_xM}+\frac{r^4}{3D(D+2)}\left(3A_{\mathcal H}-5\operatorname{Ric}^{\sharp}\right)+O(r^6),
\label{eq:curved_step_covariance}
\end{equation}
where the mean curvature vector $\mathcal H=\operatorname{tr}\mathrm{II}$ is the trace of the second fundamental form $\mathrm{II}$,  and  $\langle A_{\mathcal H}u,v\rangle
=\langle\mathrm{II}(u,v),\mathcal H\rangle$ in which $A_{\mathcal H}$ is the shape operator with respect to $\mathcal H$. Consequently,
\begin{equation}
\mathbb E\,d_M\!\left(x,\pi(\iota(x)+r\omega)\right)^2=\frac{m}{D}r^2+\frac{3|\mathcal H|^2-5\operatorname{R}}{3D(D+2)}r^4+O(r^6).
\label{eq:curved_step_distance}
\end{equation}
\end{theorem}

\begin{proof}
Write $\omega=a+n$ with $a=P_{T_xM}\omega$ and $n=P_{N_xM}\omega$.
In normal coordinates at $x$, the Gauss and Weingarten formulas
\cite[Ch.~8, pp.~134--137]{Lee1997Riemannian} give
$\iota(\exp_x\xi)=\iota(x)+\xi+\frac12\mathrm{II}(\xi,\xi)
+\frac16[(\nabla^\perp_\xi\mathrm{II})(\xi,\xi)
-A_{\mathrm{II}(\xi,\xi)}\xi]+O(|\xi|^4)$.
Expanding the closest-point condition
$\iota(x)+r\omega-\iota(\exp_x\xi_r)\perp T_{\exp_x\xi_r}M$ yields
$\xi_r=ra+r^2A_na+r^3[A_n^2a-\frac13A_{\mathrm{II}(a,a)}a
+\mathcal T_x(a,n)]+O(r^4)$, where $\mathcal T_x(a,n)$ is linear in $n$
and hence vanishes after spherical averaging.

Let $\{\nu_\alpha\}$ be an orthonormal normal frame and
$Q=\sum_\alpha A_{\nu_\alpha}^2, c_D^{-1}=D(D+2)$. Using
$\mathbb E[\omega_i\omega_j\omega_k\omega_l]
=c_D(\delta_{ij}\delta_{kl}+\delta_{ik}\delta_{jl}
+\delta_{il}\delta_{jk})$, we obtain
\begin{equation}
\begin{aligned}
\mathbb E[(A_na)\otimes(A_na)]&=c_DQ,\\
\mathbb E[a\otimes A_n^2a+A_n^2a\otimes a]&=2c_DQ,\\
\mathbb E[a\otimes A_{\mathrm{II}(a,a)}a
+A_{\mathrm{II}(a,a)}a\otimes a]
&=2c_D(A_{\mathcal H}+2Q).
\end{aligned}
\end{equation}
Therefore
$\mathbb E[\xi_r\otimes\xi_r]
=\frac{r^2}{D}I_{T_xM}
+\frac{r^4}{3D(D+2)}(5Q-2A_{\mathcal H})+O(r^6)$.
Contracting the Gauss equation gives
$\operatorname{Ric}^{\sharp}=A_{\mathcal H}-Q$, which yields
Eq.~\eqref{eq:curved_step_covariance}. The expectation is even in $r$ by
$\omega\mapsto-\omega$, so the next term is $O(r^6)$. Taking the trace and
using $\operatorname{tr}A_{\mathcal H}=|\mathcal H|^2$ and
$\operatorname{tr}\operatorname{Ric}^{\sharp}=\operatorname{R}$ gives
Eq.~\eqref{eq:curved_step_distance}.
\end{proof}

\begin{table}[t]
\centering
\caption{Uniform ambient sampling on the curved test surface ($m=2$). Theory predicts only the local squared tangent component; walk length is an end-to-end solver quantity.}
\label{tab:uniform_tangent_check}
\scriptsize
\setlength{\tabcolsep}{4.0pt}
\begin{tabular}{c c c c}
\toprule
$D$ & $m/D$ & Measured $\|P_T\omega\|^2$ & Steps / walk \\
\midrule
$3$ & $.667$ & $.664$ & $12.75$ \\
$4$ & $.500$ & $.502$ & $16.31$ \\
$6$ & $.333$ & $.336$ & $27.48$ \\
$8$ & $.250$ & $.250$ & $37.68$ \\
\bottomrule
\end{tabular}
\end{table}

% Table~\ref{tab:uniform_tangent_check} verifies the leading term of
% Theorem~\ref{thm:uniform}: uniform ambient sampling contributes only the
% fraction $m/D$ of its squared motion to the tangent space, explaining the
% growth in walk length with $D$. The $r^4$ term shows what this flat argument
% misses. The correction $-5\operatorname{Ric}^{\sharp}$ is intrinsic and
% therefore unchanged by an isometric re-embedding; directionally, it measures
% how averaged sectional curvature modifies finite-radius tangent progress.
% In contrast, $3A_{\mathcal H}$ is extrinsic and depends on how the same metric
% is embedded, coupling the mean-curvature normal to the second fundamental
% form. Hence two isometric embeddings have the same PDE and the same leading
% $m/D$ penalty, but can have different finite-radius projected-step
% covariances.

% Taking the trace makes this separation explicit:
% the correction is proportional to
% $3|\mathcal H|^2-5\operatorname{R}$, where $\operatorname{R}$ is fixed by the
% intrinsic metric while $\mathcal H$ can change under re-embedding.
% Equivalently, the Gauss equation rewrites it as
% $5|\mathrm{II}|^2-2|\mathcal H|^2$, showing that the first finite-radius
% departure from the flat law is governed directly by local embedding bending.
% Thus re-embedding can change PWoS efficiency without changing the target
% operator: it changes both the admissible geometric radius and the extrinsic
% curvature correction, whereas increasing $D$ simultaneously imposes the
% universal $m/D$ loss. This motivates biasing ambient samples toward the local
% tangent space. 

Table~\ref{tab:uniform_tangent_check} confirms the leading $m/D$ tangent fraction. The curvature correction in Theorem~\ref{thm:curved-step} separates the intrinsic Ricci term from the embedding-dependent contribution $A_{\mathcal H}$. Its trace, $(3|\mathcal H|^2-5\operatorname R)/(3D(D+2))$, shows that isometric re-embeddings can change finite-radius projected progress while preserving the target PDE. This motivates tangent-biased ambient proposals. We therefore sample $\mu\sim\mathrm{Unif}(\mathbb S^{m-1}\cap T)$ and then $\omega\mid\mu\sim\mathrm{vMF}_D(\mu,\kappa)$\cite{MardiaJupp2000DirectionalStatistics,Wood1994vMF}; marginalizing $\mu$ preserves isotropy within $T$ while allowing $\kappa$ to control how much tangent motion is recovered.

% The table makes the design tension visible: increasing $D$ removes a larger fraction of each ambient direction through projection onto the tangent plane, so geometric gains from higher codimension must compensate for slower intrinsic motion.

% Theorem~\ref{thm:uniform} exposes a cost invisible to clearance alone:
% once the admissible radius stops increasing, uniform ambient sampling loses
% intrinsic progress as $D$ grows. We therefore ask how much tangent-biased
% sampling can recover. We first sample
% $\mu\sim\mathrm{Unif}(\mathbb S^{m-1}\cap T)$ and then
% $\omega\mid\mu\sim\mathrm{vMF}_D(\mu,\kappa)$
% \cite{MardiaJupp2000DirectionalStatistics,Wood1994vMF}.
% Marginalizing $\mu$ preserves isotropy within $T$ and removes directional drift.

% \paragraph{Can we recover the discarded motion?}
% A natural graphics interpretation is to guide the ambient sample toward the tangent space, analogous to steering samples toward directions that contribute more to an estimator. Concentrating around one tangent direction would introduce drift, so we instead randomize the vMF center uniformly inside the tangent sphere. The resulting proposal remains isotropic along the surface while controlling how strongly samples avoid the normal bundle. Theorem~\ref{thm:progress} makes this control explicit through $\kappa$.

\begin{figure}
    \centering
    \includegraphics[width=\linewidth]{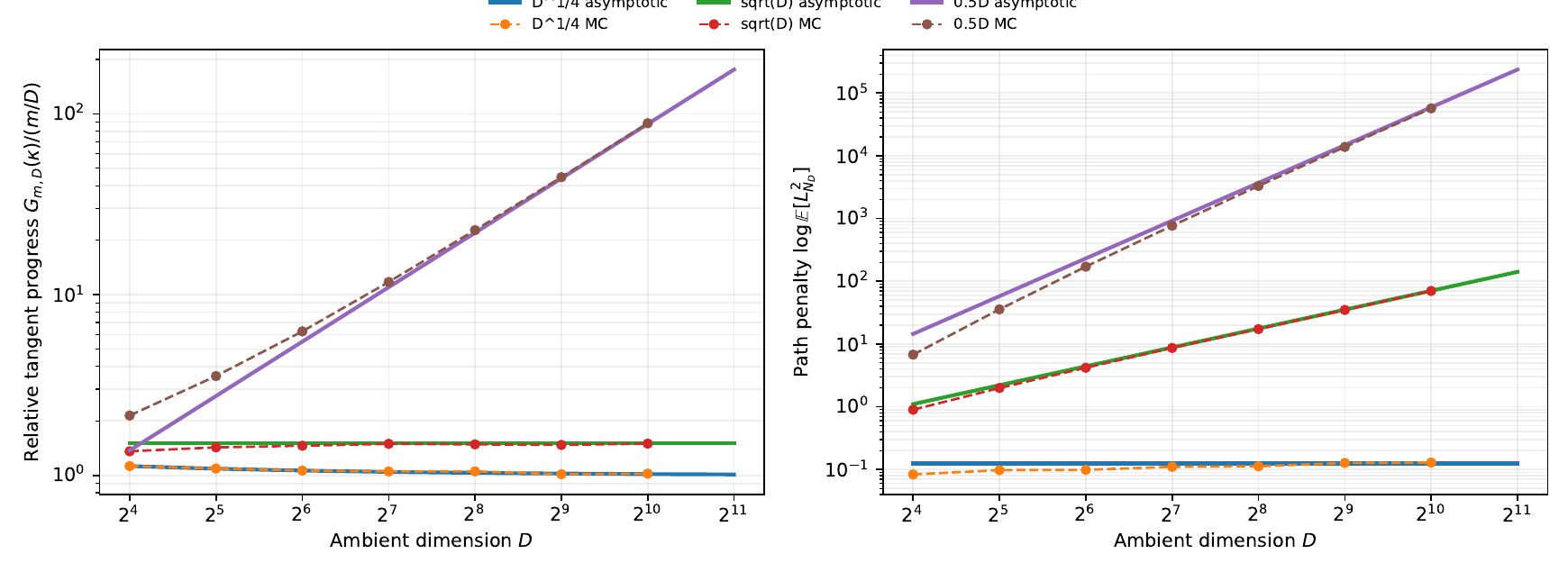}
  \caption{Large-$D$ numerical validation of Theorems~5.3--5.5. The measured tangent progress and accumulated importance-weight second moments follow the predicted asymptotic regimes, illustrating the codimension--variance barrier underlying Corollary~5.6. Markers show direct Monte Carlo estimates from $vMF$ samples and importance-ratio samples, while solid curves denote the corresponding asymptotic predictions.}
    \label{fig:theorem5_largeD_combined}
\end{figure}

\begin{theorem}[Exact tangent progress]
\label{thm:progress}
Let $q_\kappa$ denote the tangent-isotropic vMF mixture above and $A_D(\kappa)={I_{D/2}(\kappa)}/{I_{D/2-1}(\kappa)}.$ Then
\begin{equation}
G_{m,D}(\kappa):=\mathbb E_{q_\kappa}\|P_T\omega\|^2=1-(D-m){A_D(\kappa)}/{\kappa}.
\label{eq:vmf_progress}
\end{equation}
$G_{m,D}(0)=m/D$ by continuity, while for fixed $D$, $\underset{\kappa \to \infty} {\lim}G_{m,D}(\kappa)=1$.
\end{theorem}

\begin{proof}
Condition on a tangent direction $\mu\in T$.
For $\omega\sim\mathrm{vMF}_D(\mu,\kappa)$, rotational symmetry around
$\mu$ gives $\mathbb E[\omega\omega^\top\mid\mu]=a\,\mu\mu^\top+b\,(I-\mu\mu^\top).$ The standard vMF second moments satisfy $b=A_D(\kappa)/\kappa$ and $a=1-(D-1)b.$ Since $\mu\in T$, one tangent direction is parallel to $\mu$ and the remaining
$m-1$ tangent directions are orthogonal to it. Hence $\mathbb E[\|P_T\omega\|^2\mid\mu]=a+(m-1)b=1-(D-m)\frac{A_D(\kappa)}{\kappa},$ which is independent of $\mu$ and therefore remains unchanged after
marginalization. Using
$A_D(\kappa)=\kappa/D+O(\kappa^3)$ as $\kappa\to0$ gives
$G_{m,D}(0)=m/D$. For fixed $D$,
$A_D(\kappa)\to1$ as $\kappa\to\infty$, so
$A_D(\kappa)/\kappa\to0$ and $G_{m,D}(\kappa)\to1$.
\end{proof}

The quantity $G_{m,D}(\kappa)$ likewise predicts the squared component of the sampled direction in the \emph{local tangent plane}, rather than the final geodesic displacement after projection. Its finite-dimensional prediction is compared with the measured proposal statistics and the complete PDE estimator in Table~\ref{tab:finiteD_tradeoff}. Thus $\kappa$ controls how much of an ambient sample is allocated to locally useful tangent motion, but its solver benefit must still be judged together with the importance weights.

Biasing the ambient direction changes the sampling measure and therefore requires
importance weighting. Let $p_D$ denote the uniform density on $\mathbb S^{D-1}$ and let $q_\kappa$ be the tangent-isotropic vMF mixture above. After marginalizing its random tangent center, $q_\kappa(\omega)$ depends only on $\rho=\|P_T\omega\|$. We write $R_{D,\kappa}(\rho)=q_\kappa(\omega)/p_D(\omega)$ and use the correction $w_\kappa=R_{D,\kappa}^{-1}$. This marginalization avoids introducing weight variation from the auxiliary tangent-center sample itself. More importantly, the correction is multiplied along the entire walk: as in importance-sampled rendering, a proposal that produces better individual steps can still be worse if its likelihood ratios fluctuate. Theorem~\ref{thm:weights} quantifies this path-level cost.
\begin{theorem}[One-step and path-weight variance]
\label{thm:weights}
For the marginalized importance weight $w_\kappa$ such that $\mathbb E_{q_\kappa}[w_\kappa]=1,$ and
\begin{equation}
V_w(D,\kappa):=\mathbb E_{q_\kappa}[w_\kappa^2]-1=\mathbb E_{S}\left[R_{D,\kappa}(\sqrt S)^{-1}\right]-1 .
\label{eq:vmf_weight_variance}
\end{equation}
For iid flat steps, the angular likelihood ratio $L_N=\prod_{k=1}^{N}w_\kappa(\omega_k)$ satisfies $\mathbb E_{q_\kappa}[L_N^2]=\left(1+V_w(D,\kappa)\right)^N .$ For fixed $m$ and $\kappa=o(\sqrt D)$,
\begin{equation}
\begin{split}
G_{m,D}(\kappa)
&=
\frac mD+\frac{\kappa^2}{D^2}(1+o(1)), V_w(D,\kappa)=\frac{\kappa^4}{2mD^2}(1+o(1)).
\end{split}
\label{eq:subsqrt_progress}
\end{equation}
Here $L_N$ contains only the angular likelihood ratio and excludes the
screening factor appearing in $W_N$ of
Eq.~\eqref{eq:hui_wop_source_update}.
\end{theorem}

\begin{proof}
For the marginalized vMF mixture, the likelihood ratio has the closed form
$R_{D,\kappa}(\rho)=\frac{C_D(\kappa)}{C_D(0)}
\mathcal M_m(\kappa\rho)$, where
$\mathcal M_m(z)=\Gamma(m/2)\left(\frac{2}{z}\right)^{m/2-1}I_{m/2-1}(z)$.
Since $w_\kappa=p_D/q_\kappa$,
$\mathbb E_{q_\kappa}[w_\kappa]=1$ and
$\mathbb E_{q_\kappa}[w_\kappa^2]
=\mathbb E_{p_D}[R_{D,\kappa}(\|P_T\omega\|)^{-1}]$.
Theorem~\ref{thm:uniform} gives
$S=\|P_T\omega\|^2\sim\mathrm{Beta}(m/2,(D-m)/2)$, proving
Eq.~\eqref{eq:vmf_weight_variance}. Independence similarly gives
$\mathbb E[L_N^2]=(1+V_w)^N$.

For $\kappa=o(\sqrt D)$, the vMF mean resultant length satisfies
$A_D(\kappa)=\frac{\kappa}{D}
-\frac{\kappa^3}{D^2(D+2)}
+O\!\left(\frac{\kappa^5}{D^5}\right)$, and substitution into
Eq.~\eqref{eq:vmf_progress} yields
$G_{m,D}(\kappa)=\frac mD+\frac{\kappa^2}{D^2}(1+o(1))$.

For the weight variance, use
$\mathbb E[S^j]=(\frac m2)_j/(\frac D2)_j=O(D^{-j})$ for fixed $m$.
The power series of $\mathcal M_m$ gives
$\mathcal M_m(z)=1+\frac{z^2}{2m}
+\frac{z^4}{8m(m+2)}+O(z^6)$, hence
\begin{equation}
\mathbb E_S\!\left[\mathcal M_m(\kappa\sqrt S)^{-1}\right]
=
1-\frac{\kappa^2}{2D}
+
\frac{m+4}{8mD(D+2)}\kappa^4
+
o\!\left(\frac{\kappa^4}{D^2}\right).
\end{equation}
Likewise,
\begin{equation}
\frac{C_D(0)}{C_D(\kappa)}
=
1+\frac{\kappa^2}{2D}
+
\frac{\kappa^4}{8D(D+2)}
+
o\!\left(\frac{\kappa^4}{D^2}\right).
\end{equation}
Multiplying these expansions, the quadratic terms cancel and
\begin{equation}
V_w(D,\kappa)
=
\frac{D-m}{2mD^2(D+2)}\kappa^4
+
o\!\left(\frac{\kappa^4}{D^2}\right)
=
\frac{\kappa^4}{2mD^2}(1+o(1)).
\end{equation}
\end{proof}
\begin{table}[t]
\centering
\caption{Finite-$D$ check at $D=8$. ``th./emp.'' reports the theoretical and measured local squared tangent component from the existing full PWoS sweep.}
\label{tab:finiteD_tradeoff}
\scriptsize
\setlength{\tabcolsep}{3.0pt}
\begin{tabular}{c c c c c}
\toprule
$\kappa$ & $G$ th./emp. & $\mathbb E[w^2]$ th./emp. & ESS & RMSE \\
\midrule
$0$   & $.250/.250$ & $1.000/1.000$ & $1.000$ & $.0094$ \\
$0.5$ & $.252/.253$ & $1.000/1.000$ & $.995$ & $.0075$ \\
$1$   & $.259/.260$ & $1.002/1.002$ & $.918$ & $.0164$ \\
$2$   & $.285/.281$ & $1.029/1.035$ & $.472$ & $.1464$ \\
$4$   & $.363/.366$ & $1.307/1.286$ & $.121$ & $.3120$ \\
\bottomrule
\end{tabular}
\end{table}

Table~\ref{tab:finiteD_tradeoff} illustrates this tradeoff at $D=8$: stronger tangent guidance increases local progress and shortens walks, but reduces path ESS and eventually degrades RMSE.

\subsection{Dimension Scaling and Selection}
The finite-$D$ results suggest using only mild tangent guidance, but they leave a more fundamental design question: could $\kappa$ grow with $D$ so that tangent bias cancels the $m/D$ progress loss and makes arbitrarily high-dimensional embeddings cheap? Theorem~\ref{thm:barrier} tests exactly this possibility. Its three regimes correspond to increasingly aggressive attempts to recover tangent motion as codimension grows.

\begin{theorem}[Codimension--variance barrier]
\label{thm:barrier}
Fix $m,r$ in the iid flat model and let $D\to\infty$.

\textbf{(i) Subcritical:} if $\kappa=o(\sqrt D)$ and $N_D\sim D/m$,
\begin{align}
\log\mathbb E_q[L_{N_D}^2]
&=
\frac{\kappa^4}{2m^2D}(1+o(1)),
\label{eq:path_subsqrt}\\
\frac{G_{m,D}(\kappa)-m/D}{m/D}
&=
O(D^{-1/2})
\qquad
\text{when }\kappa=O(D^{1/4}),
\label{eq:stable_bias_no_gain}
\end{align}
so within the subcritical regime, uniformly bounded path second moments allow only asymptotically vanishing relative tangent-progress gain.

\textbf{(ii) Constant-factor gain:} if $\kappa=\beta\sqrt D$, $\beta>0$,
\begin{align}
G_{m,D}(\kappa)
&=
\frac{m+\beta^2}{D}+o(D^{-1}),
\label{eq:sqrtD_progress}\\
\lim_{D\to\infty}V_w(D,\kappa)
&=
V_m(\beta)
:=
e^{\beta^2/2}
\mathbb E_{X\sim\chi_m^2}
\left[
\mathcal M_m(\beta\sqrt X)^{-1}
\right]-1
>0.
\label{eq:sqrtD_variance}
\end{align}
Hence $N_D=\Theta(D)$ gives $\mathbb E_q[L_{N_D}^2]=\exp(\Theta(D))$: a constant-factor improvement in tangent progress is accompanied by exponentially growing angular path second moments.

\textbf{(iii) Order-one gain:} if $\kappa=\alpha D$, $\alpha>0$, let
$a(\alpha)=(\sqrt{1+4\alpha^2}-1)/(2\alpha)$ and
$\psi(\alpha)=\alpha a(\alpha)+\frac12\log(1-a(\alpha)^2)$. Then
\begin{align}
\lim_{D\to\infty}G_{m,D}(\alpha D)
&=
a(\alpha)^2,
\label{eq:linearD_progress}\\
\lim_{D\to\infty}
\frac1D\log\!\left(1+V_w(D,\alpha D)\right)
&=
\psi(\alpha)>0.
\label{eq:linearD_variance}
\end{align}
Thus achieving order-one tangent progress in this linear-concentration regime incurs an exponential one-step importance-weight penalty. These statements characterize the three specified scalings and do not claim to exhaust all intermediate choices of $\kappa(D)$.
\end{theorem}
\begin{proof}
For (i), Theorem~\ref{thm:weights} gives
$V_w(D,\kappa)=\frac{\kappa^4}{2mD^2}(1+o(1))$. Using
$\log(1+x)=x+o(x)$ and $N_D\sim D/m$,
\[
\log\mathbb E[L_{N_D}^2]
=
N_D\log(1+V_w)
=
\frac{\kappa^4}{2m^2D}(1+o(1)),
\]
which gives Eq.~\eqref{eq:path_subsqrt}. Boundedness requires
$\kappa^4/D=O(1)$, hence $\kappa=O(D^{1/4})$. Equation~\eqref{eq:subsqrt_progress} then gives
\[
\frac{G_{m,D}(\kappa)-m/D}{m/D}
=
\frac{\kappa^2}{mD}(1+o(1))
=
O(D^{-1/2}),
\]
and therefore the relative gain tends to zero.

For (ii), set $\kappa=\beta\sqrt D$. The large-order Bessel-ratio expansion
$\frac{A_D(\beta\sqrt D)}{\beta\sqrt D}
=\frac1D-\frac{\beta^2}{D^2}+o(D^{-2})$
inserted into Eq.~\eqref{eq:vmf_progress} gives
Eq.~\eqref{eq:sqrtD_progress}. By Theorem~\ref{thm:uniform},
$DS\overset{d}{\to}X$ with $X\sim\chi_m^2$. At the same scaling,
$\frac{C_D(\beta\sqrt D)}{C_D(0)}$ converges to $e^{-\beta^2/2}$, while continuity of $\mathcal M_m$ gives
$\mathcal M_m(\beta\sqrt{DS})
\overset{d}{\to}
\mathcal M_m(\beta\sqrt X)$.
Dominated convergence in Eq.~\eqref{eq:vmf_weight_variance} yields
Eq.~\eqref{eq:sqrtD_variance}. Since this limit is strictly positive,
$N_D=\Theta(D)$ implies
$\mathbb E[L_{N_D}^2]=(1+V_w)^{N_D}=\exp(\Theta(D))$.

For (iii), uniform large-order asymptotics of the modified Bessel functions give
$\lim_{D\to\infty}A_D(\alpha D)=a(\alpha)$, where
$\alpha(1-a^2)=a$. Hence
\[
\lim_{D\to\infty}G_{m,D}(\alpha D)
=
1-\frac{a(\alpha)}{\alpha}
=
a(\alpha)^2,
\]
giving Eq.~\eqref{eq:linearD_progress}. Laplace's method applied to the spherical integral defining the vMF normalizer gives
\[\lim_{D\to\infty}-\frac1D\log\frac{C_D(\alpha D)}{C_D(0)}=\psi(\alpha)=\alpha a(\alpha)+\frac12\log(1-a(\alpha)^2).\]
The remaining factor $\mathcal M_m(\alpha D\sqrt S)$ has fixed tangent dimension $m$ and does not cancel this exponential normalization rate. Substitution into Eq.~\eqref{eq:vmf_weight_variance} gives Eq.~\eqref{eq:linearD_variance}; since $a(\alpha)\in(0,1)$ for $\alpha>0$, $\psi(\alpha)>0$.
\end{proof}

\begin{table}[t]
\centering
\caption{Geometry--transport tradeoff after re-embedding. Residual projection conflicts improve with dimension, while uniform-sampling walks become longer.}
\label{tab:dimension_geometry_tradeoff}
\scriptsize
\setlength{\tabcolsep}{5.0pt}
\begin{tabular}{c c c}
\toprule
$D$ & Residual face conflicts & Steps / walk \\
\midrule
$3$ & $4531$ & $14.15$ \\
$4$ & $546$  & $19.31$ \\
$6$ & $340$  & $29.88$ \\
$8$ & $1$    & $38.77$ \\
\bottomrule
\end{tabular}
\end{table}

For a direction-independent walk radius $r_D(x)$ chosen before sampling
$\omega$, define the expected intrinsic squared progress per step as
\begin{equation}
\mathcal J_{D,\kappa}
=
\mathbb E[r_D(x)^2]\,G_{m,D}(\kappa).
\label{eq:dimension_efficiency}
\end{equation}
For uniform sampling, $G_{m,D}(0)=m/D$, so $\mathcal J_{D,0}$ reduces to the
expected squared radius multiplied by the tangent fraction.

\begin{corollary}[Finite-dimension design implication]
\label{cor:dimension}
Assume the admissible walk radii satisfy
\begin{equation}
\sup_D \mathbb E[r_D(x)^2] \le C_r < \infty .
\end{equation}
Then
\begin{equation}
\mathcal J_{D,0}
=
\frac mD\mathbb E[r_D(x)^2]
\le
\frac{mC_r}{D},
\end{equation}
so uniform sampling requires $\Omega(D)$ local steps for fixed intrinsic
mean-square progress. Within the three regimes of
Theorem~\ref{thm:barrier}, avoiding exponential angular second-moment growth
does not improve this asymptotic step order.
\end{corollary}

\begin{proof}
The uniform case follows immediately from
$G_{m,D}(0)=m/D$ and the bounded-radius assumption. In
Theorem~\ref{thm:barrier}(i), bounded angular path second moments require
$\kappa=O(D^{1/4})$, for which
$G_{m,D}(\kappa)=\frac mD(1+o(1))$; hence $\mathcal J_{D,\kappa}=\mathcal O(D^{-1})$, so the
$\Omega(D)$ lower bound on the number of local steps remains.
In regime (ii),
$G_{m,D}(\kappa)=(m+\beta^2)D^{-1}+o(D^{-1})$, so only the constant changes,
while $\mathbb E_q[L_{N_D}^2]=\exp(\Theta(D))$. In regime (iii),
$G_{m,D}(\alpha D)$ has an order-one limit, but
$1+V_w(D,\alpha D)=\exp(D\psi(\alpha)+o(D))$. Thus improving the asymptotic
progress order within these regimes requires exponential importance-weight
growth.
\end{proof}

Theorem~\ref{thm:barrier} and Corollary~\ref{cor:dimension} therefore suggest
using $\mathcal J_{D,\kappa}$ together with $V_w(D,\kappa)$: the former predicts
transport progress, while the latter rejects proposals whose shorter walks are
obtained at excessive importance-weight variance. In practice,
$\mathbb E[r_D(x)^2]$ can be estimated from a small set of representative
states or a short pilot, while $G_{m,D}(\kappa)$ and $V_w(D,\kappa)$ are
analytic. Hence only a small number of competitive configurations require
solver-level evaluation rather than an exhaustive PDE sweep.

We verify the progress model on the complete $(D,\kappa)$ sweep. If
$\mathcal J_{D,\kappa}$ measures intrinsic squared progress per step, then
$\mathrm{Steps}\propto1/\mathcal J_{D,\kappa}$, equivalently
$\mathrm{Steps}\cdot\mathcal J_{D,\kappa}$ should remain approximately
constant. Across all 20 configurations, $1/\mathcal J_{D,\kappa}$ correlates with the
measured walk length at $0.992$, while
$\mathrm{Steps}\cdot\mathcal J_{D,\kappa}$ varies by only $4.95\%$. The agreement becomes slightly stronger after excluding strongly ill-conditioned importance weights, supporting $\mathcal J_{D,\kappa}$ as a
predictor of transport cost while leaving estimator stability to
$V_w(D,\kappa)$.

\begin{table}[t]
\centering
\caption{Validation of $\mathcal J_{D,\kappa}$ as a transport-cost predictor.
We compare its inverse with measured steps per walk; CV measures the relative
variation of $\mathrm{Steps}\cdot\mathcal J_{D,\kappa}$, and ESS is the
normalized effective sample size of the accumulated path importance weights.}
\label{tab:progress_proxy_validation}
\small
\begin{tabular}{lcc}
\toprule
Configurations &
$\operatorname{corr}(1/\mathcal J,\mathrm{Steps})$ &
CV$(\mathrm{Steps}\cdot\mathcal J)$ \\
\midrule
All 20 & 0.992 & 4.95\% \\
ESS$>0.9$ & 0.994 & 4.29\% \\
\bottomrule
\end{tabular}
\end{table}

\section{Experiments}
\label{sec:experiments}
We evaluate WoAP from three perspectives: whether an improved geometric representation accelerates the PDE solver without sacrificing accuracy, how re-embedding and projection states control branch selection, and how ambient dimension and angular sampling jointly determine computational cost.

\subsection{End-to-End Accuracy and Walk Efficiency}
We first compare WoAP with standard PWoS on a developable folded strip, where an exact flat isometry isolates the effect of the representation from embedding-optimization errors. The intrinsic PBM reference is estimated by $10^5$ planar 2D WoS paths per query on the unfolded domain, with RMSE $5.02\times10^{-4}$ against the analytic harmonic solution. An independent $32768$-path/query PWoS estimate reaches RMSE $1.06\times10^{-3}$ against the same solution. Standard PWoS uses the spatially varying local-feature-size radius $r=\min\{\delta(x),\ell(x)\}$. We also test a large-radius baseline that removes this restriction while retaining global closest-point projection, and WoAP using the exact isometric representation.

% \begin{table}[t]
% \centering
% \caption{End-to-end comparison on the developable folded strip. RMSE is mean$\pm$std over 10 independent runs with 256 paths/query against the $10^5$-path/query intrinsic 2D WoS (PBM) reference. Wrong projections violate the local-branch displacement bound.}
% \label{tab:pwos_e2e}
% \small
% \setlength{\tabcolsep}{3pt}
% \begin{tabular}{lccc}
% \toprule
% Method & RMSE & Steps/path & \shortstack{Wrong proj.\\/path}\\
% \midrule
% PWoS & $0.00734\pm0.00223$ & $295.68\pm4.87$ & 0\\
% Large-radius & $0.14734\pm0.00432$ & $15.28\pm0.22$ & $4.99\pm0.07$\\
% \textsc{WoAP} & $0.00734\pm0.00125$ & $15.96\pm0.27$ & 0\\
% \bottomrule
% \end{tabular}
% \end{table}

\begin{table}[t]
\centering
\caption{Folded-strip end-to-end comparison against the $10^5$-path/query intrinsic PBM reference. RMSE and solver time are mean$\pm$std over 10 runs with 256 paths/query. Time excludes preprocessing and JIT compilation; remote jumps exceed $2.05r$ in intrinsic displacement.}
\label{tab:pwos_e2e}
\small
\setlength{\tabcolsep}{2pt}
\begin{tabular}{@{}lcccc@{}}
\toprule
Method & \shortstack{RMSE\\($10^{-3}$)} & Steps & \shortstack{Remote\\jumps} & \shortstack{Time\\(ms/query)}\\
\midrule
PWoS & $7.34\pm2.23$ & $295.68$ & $0$ & $20.92\pm0.37$\\
Large-radius & $147.34\pm4.32$ & $15.28$ & $4.99$ & $1.23\pm0.02$\\
\textsc{WoAP} & $7.34\pm1.25$ & $15.96$ & $0$ & $0.48\pm0.01$\\
\bottomrule
\end{tabular}
\end{table}

Table~\ref{tab:pwos_e2e} shows that WoAP matches the error level of PWoS while reducing the average number of stochastic steps by $18.5\times$. The large-radius baseline achieves similarly short walks but incurs substantial error, indicating that increasing the step radius alone is insufficient without reliable branch selection. We identify wrong-sheet events using a curvature-based displacement test: with measured $K_{\max}=2.635$, a connected local projection is expected to satisfy $d_M(x,y)\le 2K_{\max}^{-1}\arcsin(K_{\max}r)$. At the largest tested radius $r=0.18$, the bound is $2.084r$. The large-radius baseline exceeds it approximately five times per path, whereas no violations are observed for PWoS or WoAP. This comparison establishes the potential solver benefit of an exact isometry; the optimized re-embedding and projection-state components are evaluated next.

\subsection{Representation Quality and Branch Tracking}
We study re-embedding on a five-turn warped Swiss roll with 360 vertices, 624 triangles, inter-layer gap $0.065$, and target radius $R^*=0.14$. The embedding sweep uses the PBD-style edge projection~\cite{Muller2007PBD} and measures nonlocal conflicts as face pairs with $R_{\rm sym}<R^*$, excluding two-ring neighbors after a $2R^*$ centroid broad phase. To evaluate the runtime projection policy, we compare each endpoint with a 12-substep same-sample continuation reference and record the fraction whose discrepancy exceeds $10^{-3}$. Projection cost includes the interior Green query and continuation substeps.

\begin{table}[t]
\centering
\caption{Embedding and branch-tracking statistics on the warped Swiss roll. Mismatch denotes endpoint disagreement with 12-substep continuation.}
\label{tab:main_geometry}
\small
\begin{tabular}{lrrrr}
\toprule
Embedding & Conflicts & Transition & Mismatch & Proj./walk\\
\midrule
Input 3D & 5030 & 21.8\% & 23.67\% & 76.65\\
Opt. 3D & 4531 & 28.2\% & 11.00\% & 45.21\\
Opt. 4D & 546 & 10.3\% & 0.67\% & 39.29\\
Opt. 6D & 340 & 15.9\% & 1.67\% & 59.33\\
Opt. 8D & 1 & 5.1\% & 0.00\% & 74.50\\
\bottomrule
\end{tabular}
\end{table}

Re-embedding strongly suppresses nonlocal competition (Table~\ref{tab:main_geometry}): conflicts decrease from $5030$ in the input geometry to $546$ at $D=4$ and only one at $D=8$. Endpoint mismatch falls from $23.67\%$ to $0.00\%$, while the transition set occupies just $5.1\%$ of faces at $D=8$. However, projection cost is lowest at $D=4$, demonstrating that greater geometric clearance does not always produce faster walks. Figure~\ref{fig:main_embedding} visualizes the embedding changes and corresponding transition maps. The separate 48-vertex hard-constraint pilot reduces conflicts from $309$ at $D=3$ to $105$--$106$ at $D=6$--$8$, with maximum edge error below $3.34\times10^{-10}$; the large-mesh results above use the approximate PBD-style optimizer.

\begin{figure}[t]
\centering
\includegraphics[width=0.9\linewidth]{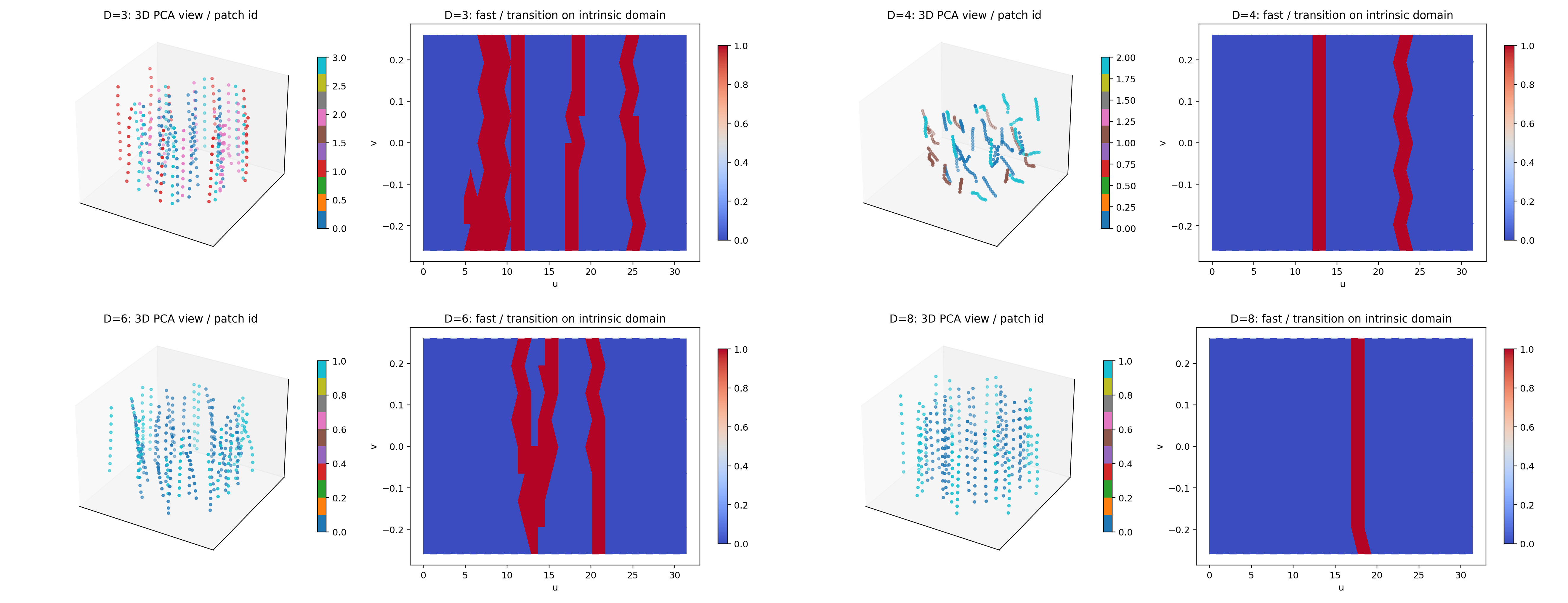}
\caption{Higher-dimensional embeddings and their intrinsic-domain transition maps. Increased separation reduces competing projection branches but changes continuation workload.}
\label{fig:main_embedding}
\Description{High-dimensional embeddings and transition maps.}
\end{figure}

To examine the projection-state mechanism independently, we use a finer Swiss roll with 1,408 vertices and 2,610 triangles. At $K=5$, increasing $\lambda_{\rm conf}$ from $0$ to $5$
with the regularization weights fixed reduces weighted same-owner
conflict from $0.1168$ to $0.0366$, while interface length
increases from $3.69$ to $6.46$. With $K=6$, the measured same-owner conflict vanishes, all patches remain connected, and $15.21\%$ of faces enter the transition set. Selective continuation then uses about $2.65$ local projections per step, compared with $12$ for the fine continuation reference, with approximately $0.2\%$ endpoint disagreement. In a separate 256-walk/query harmonic test, patch tracking takes $14.2$ stochastic steps and $21.2$ projections per walk, compared with $186.6$ steps for the conservative global-medial baseline. A global large-radius baseline incurs about $3.96$ remote jumps per walk, whereas none are observed with patch tracking. Together, these measurements show how separating competing sheets and continuing only near transition regions reduce branch-handling cost.

\subsection{Dimension--Sampling Tradeoff}
Reducing projection conflicts favors higher ambient dimension, but projected random steps lose tangent progress as $D$ increases. We quantify this tradeoff on the warped Swiss roll by solving $-\Delta_Mu+2u=f$ with nonzero source and Dirichlet data. We test $D\in\{3,4,6,8\}$ and tangent-mixture concentrations $\kappa\in\{0,0.5,1,2,4\}$, using 64 walks at each of ten fixed queries per configuration. RMSE is measured against a $40\times9$ surface-FEM solution. An independent $20\times5$ surface-FEM predictor (RMSE $1.06\times10^{-3}$ against this reference) supplies the normal-Taylor compensation $\mathfrak m^*$; the reference field does not enter the estimator. We report stochastic steps, total closest-point queries per walk, and the effective sample-size (ESS) fraction of accumulated angular importance weights.

\begin{table}[t]
\centering
\caption{Selected screened-Poisson results on the warped Swiss roll (64 walks/query). RMSE uses the fine surface-FEM reference; ESS measures accumulated angular importance weights.}
\label{tab:main_sampling}
\small
\begin{tabular}{rrrrrr}
\toprule
$D$ & $\kappa$ & Steps & Proj./walk & ESS & RMSE\\
\midrule
3 & 0 & 12.75 & 57.73 & 1.000 & 0.0406\\
3 & 0.5 & 12.51 & 53.10 & 0.996 & 0.0389\\
4 & 0 & 16.31 & 55.85 & 1.000 & 0.0558\\
6 & 0 & 27.48 & 76.68 & 1.000 & 0.0142\\
8 & 0 & 37.68 & 76.72 & 1.000 & 0.0094\\
8 & 0.5 & 34.51 & 69.87 & 0.995 & 0.0075\\
8 & 4 & 27.62 & 55.50 & 0.121 & 0.3120\\
\bottomrule
\end{tabular}
\end{table}

The observed tangent fraction agrees with the predicted angular statistics: at $(D,\kappa)=(6,1)$, theory gives $0.3468$ and simulation gives $0.3485$. Stronger tangent bias shortens walks but also increases the variability of importance weights. At $D=8$, increasing $\kappa$ from $0$ to $4$ reduces the mean number of steps from $37.68$ to $27.62$, yet ESS falls from $1$ to $0.121$ and RMSE rises from $0.0094$ to $0.3120$ (Table~\ref{tab:main_sampling}). The full $D\times\kappa$ sweep and a separate test of the compensation term appear in Appendix~\ref{app:experiments}.

\begin{table}[t]
\centering
\caption{Fastest tested configurations at each target RMSE, estimated by bootstrap resampling of the 64-walk/query pools. Times exclude preprocessing.}
\label{tab:main_fixed}
\small
\begin{tabular}{rrrrrr}
\toprule
Target & $D$ & $\kappa$ & Walks/query & Time [s] & RMSE\\
\midrule
0.050 & 3 & 0.5 & 8 & 1.217 & 0.0493\\
0.030 & 6 & 0 & 8 & 1.759 & 0.0270\\
0.020 & 8 & 0 & 16 & 3.888 & 0.0197\\
\bottomrule
\end{tabular}
\end{table}

We finally compare configurations at fixed accuracy rather than at fixed path count. Bootstrap resampling~\cite{Efron1979Bootstrap} of the original 64-walk pools estimates RMSE and sampling time at 8, 16, 32, and 64 walks/query. Table~\ref{tab:main_fixed} selects the fastest tested configuration meeting each target: $D=3$ for RMSE $0.05$, $D=6$ for $0.03$, and $D=8$ for $0.02$. The preferred dimension therefore depends on the requested accuracy, balancing geometric clearance, tangent progress, and sampling variance rather than optimizing any one of them in isolation.

\section{Discussion and Limitations}
\label{sec:limitations}
WoAP makes extrinsic geometry an explicit design variable for an intrinsic stochastic PDE problem. Its three components address different costs: re-embedding reduces nonlocal conflicts, projection states resolve remaining branches, and the sampling analysis discourages selecting dimension from clearance alone. The current evidence is strongest for this mechanism and the measured geometry--sampling tradeoff.

Several limits remain. The main geometry and PDE sweep uses approximate edge projection; only the small pilot validates the hard-constrained optimizer. Triangle edge preservation does not make mesh closest-point maps globally smooth, and finite-direction preprocessing cannot certify every runtime branch. In particular, unchanged angular samples do not by themselves prove that a continued multi-branch extension satisfies the ambient Green identity on a large ball. Interior seed-local projection, truncated compensation, boundary stopping, and capped walks add numerical error. The theory concerns a flat fixed-radius model and angular likelihood moments, not a universal PDE-variance bound. Finally, the quantitative compensated sweep covers one folded geometry and ten queries, the accordion example is qualitative, and the fixed-RMSE estimates exclude preprocessing. Larger hard-constrained benchmarks, local-extension certification, and independent accuracy/cost trials are needed to establish broader end-to-end performance.
\label{sec:main_end}

\bibliographystyle{ACM-Reference-Format}
\bibliography{references}

\appendix
\section{Implementation Details}
\label{app:implementation}
\subsection{Embedding and Local Geometry}
The main geometry and solver experiments use the scalable PBD-style edge projection of Eq.~\eqref{eq:pbd_edge_projection}. After each repulsive update, we apply Jacobi edge-projection sweeps to restore the target intrinsic edge  lengths. Although this does not enforce the coupled edge constraints exactly, the resulting metric distortion is small in the tested embeddings: the maximum
edge-length error is $2.30\times10^{-4}$ at $D=3$,
$1.77\times10^{-4}$ at $D=4$, $5.45\times10^{-6}$ at $D=6$, and
$3.58\times10^{-7}$ at $D=8$. The corresponding RMS errors are at most $1.59\times10^{-5}$.

The conflict candidate set is rebuilt periodically from the current embedding using the centroid broad phase, while tangent-point quantities are recomputed from the current vertex positions at every optimization step. After embedding optimization, the residual conflict graph used for phase construction is recomputed from the final geometry.

\subsection{Predictor-Based Compensation}
\label{app:backend}
The Green recursion, its dimension dependence, and the use of predictor-based compensation are attributed to Hui et al.~\cite{Hui2026MPWoS}. We retain their normal Taylor construction. In codimension one, write $z=y+\eta n(y)$ and define the residual of a coarse numerical predictor $u^*$:
\begin{equation}
R^*(s;y)
=
-\Delta E u^*(y+s n(y))
+\sigma u^*(y)-f(y).
\label{eq:hui_predictor_residual}
\end{equation}
Using $\mathfrak m(y)=0$, the normal expansion is truncated as
\begin{equation}
\mathfrak m^*(x)
=
\eta \mathfrak m_{1,*}(y)
+
\frac{\eta^2}{2}\mathfrak m_{2,*}(y),
\label{eq:hui_mstar}
\end{equation}
with central differences
\begin{equation}
\begin{split}
\mathfrak m_{1,*}(y)&=\frac{R^*(h_\eta;y)-R^*(-h_\eta;y)}{2h_\eta},\\
\mathfrak m_{2,*}(y)&=\frac{R^*(h_\eta;y)-2R^*(0;y)+R^*(-h_\eta;y)}{h_\eta^2}.
\end{split}
\label{eq:hui_mstar_coeff}
\end{equation}
For codimension $q>1$, Hui et al. replace the scalar normal coordinate by local
normal coordinates and use the corresponding multivariate Taylor expansion.
Writing $\eta\in\mathbb R^q$ in a local normal frame, our mesh implementation
stores
\begin{equation}
\mathfrak m^*(y+\mathbf N_y\eta)
=
\mathbf g_y^\top\eta
+
\frac12\eta^\top\mathbf H_y\eta,
\label{eq:hui_mstar_multinormal}
\end{equation}
where the full symmetric normal Hessian $\mathbf H_y$ includes the mixed
normal derivatives.  Directly differentiating a piecewise-linear triangle
closest-point map is unstable, so we reconstruct a quadratic surface jet and
a quadratic coarse-predictor jet from the two-ring neighborhood, differentiate
the resulting smooth local closest-point extension, and barycentrically
interpolate $(\mathbf g_y,\mathbf H_y)$ at runtime. %  This is our numerical realization of Hui et al.'s multivariate normal Taylor construction, not a new compensation model.

The Green density in Eq.~\eqref{eq:hui_local_recursion} is sampled by a tabulated radial distribution and an independent uniform ambient direction. Both the radius-dependent survival factor and interior normalization use the actual ambient dimension. The fine FEM reference is never used to construct $\mathfrak m^*$ in the recursive experiment. The isolated compensation diagnostic uses it only for the boundary expectation, explicitly isolating the source-correction error.
\label{eq:hui_green_density}\label{eq:ambient_pde}

\subsection{Preprocessing and Runtime Sequence}
\label{app:runtime}
Preprocessing optimizes the embedding, constructs the conflict graph and phases, assigns face owners, and tests a band near the seams by deterministic seeded continuation. Seam faces are always marked transition; detected risk is dilated by one face ring. The complement is used for the fast endpoint rule. The map is a finite-direction test and is not a certificate. Spherical-Fibonacci directions are used for the three-dimensional diagnostic~\cite{Keinert2015SphericalFibonacci}; higher-dimensional tests use the directional sets provided by the corresponding implementation.

The runtime sequence is:
\begin{enumerate}
\item Initialize $S=0$, $W=1$, the start point, its seed face $\tau$, and owner $a_\tau$.
\item If the interpolated physical boundary-distance estimate exceeds the stopping tolerance, choose $\delta =0.1$ in $r=\min\{(1-\delta)\widehat d_{\partial M},R^*\}$ and draw an interior point from the $D$-dimensional Green density.
\item Evaluate the interior source and $\mathfrak m^*$ with a two-ring seed-local projection; add their contribution with the incoming $W$.
\item Draw one uniform or tangent-mixture vMF boundary direction. Evaluate the endpoint with the owner projector on fast faces or deterministic seeded continuation on transition faces.
\item Multiply $W$ by one survival factor and one marginalized angular importance factor. Update the point, seed face, and landing-face owner; repeat.
\item Return $S+Wb$ at the stopping state. The full-sweep code caps walks at 260 steps and applies its terminal payoff at a capped state; this is an additional numerical approximation, not an exact stopping identity.
\end{enumerate}

The compensated sweep uses eight continuation substeps, while endpoint diagnostics use a finer 12-substep numerical reference. Patch transitions preserve the ongoing stochastic walk, with owner updates handled by local continuation. Interior Green samples use seed-local projection, while the fast/transition map governs boundary endpoint projection.

% The compensated sweep uses eight continuation substeps; the endpoint diagnostic uses a 12-substep reference. The latter is a numerical oracle, not an exact branch solution. No interface restarts the walk. Interior Green queries use a seed-local evaluation in the supplied implementation; their branch validity is not established by the boundary fast/transition classifier.

\paragraph{Importance diagnostics.}
For path $s$, let $L_s=\prod_kw_\kappa(\omega_k)$. The recorded ESS fraction is $(\sum_sL_s)^2/(N\sum_sL_s^2)$, averaged over query points. It excludes the screening product $\prod_kc_{\sigma,r_k}$, matching the supplied code. % This corrects the previous draft's description of ESS as using the complete screened path weight.

\section{Supplementary Experimental Details}
\label{app:experiments}
The main text reports the geometry sweep, hard-isometry pilot, phase-domain ablation, end-to-end PWoS comparison, selected sampling configurations, and fixed-RMSE results. This appendix documents the evaluation conventions, compensation validation, complete geometry--sampling sweep, curvature-step validation, and additional qualitative geometry tests. 
\subsection{Evaluation conventions}
For the compensated screened-Poisson sweep, \emph{Steps/walk} counts stochastic
PWoS steps, whereas \emph{Proj./walk} counts all local closest-point solves,
including the interior Green query and any deterministic continuation
substeps. The reported ESS is computed from the product of angular importance
weights only; screening factors are excluded. RMSE is measured at the ten
fixed query points against the fine surface-FEM reference used in the main
text. These conventions match the released evaluation scripts.

\subsection{Isolated coarse-predictor compensation}
\label{sec:exp_compensation}
The runtime uses the normal Taylor compensation $\mathfrak m^*$ described in
Appendix~\ref{app:backend}. To verify this numerical component independently
of full-path effects, we use a cylindrical strip and a screened-Poisson
problem with a $96\times25$ surface-FEM reference and an independent
$24\times8$ predictor. The predictor has RMSE $2.25\times10^{-3}$ against the
fine field. At 24 held-out points, we evaluate a single Green step with
$10^5$ common-random-number samples and compare $\mathfrak m^*=0$ with the
coarse-predictor compensation.

\begin{table}[t]
\centering
\caption{Isolated Green-step RMSE with and without coarse-predictor
compensation.}
\label{tab:compensation}
\small
\begin{tabular}{rccc}
\toprule
$D$ & $\mathfrak m^*=0$ & Hui $\mathfrak m^*$ & Reduction \\
\midrule
3 & $5.01\times10^{-4}$ & $1.89\times10^{-4}$ & $2.64\times$ \\
4 & $3.02\times10^{-4}$ & $1.44\times10^{-4}$ & $2.10\times$ \\
6 & $1.87\times10^{-4}$ & $1.34\times10^{-4}$ & $1.39\times$ \\
8 & $1.67\times10^{-4}$ & $1.15\times10^{-4}$ & $1.45\times$ \\
\bottomrule
\end{tabular}
\end{table}

The compensation consistently reduces the isolated local Green-step error.
In the complete stochastic recursion its effect is smaller and can be masked
by projection, stopping, and Monte Carlo errors; accordingly, the main text
does not claim a full-path variance reduction from this component.

\subsection{Complete geometry--sampling sweep}
The main text shows representative rows from the compensated
$D\times\kappa$ experiment. Table~\ref{tab:vmf_tradeoff} gives the complete
20-cell sweep. Each cell uses the same geometry, owner map, ten query points,
and 64 walks/query for its ambient dimension; only the angular proposal varies
with $\kappa$.

\begin{table*}[!t]
\centering
\caption{Complete compensated $D\times\kappa$ sweep. $G$ is the expected
squared tangent component, and ESS uses accumulated angular likelihood ratios
only.}
\label{tab:vmf_tradeoff}
\scriptsize
\begin{tabular}{rrrrrrrrr}
\toprule
$D$ & $\kappa$ & $G_{\rm th}$ & $G_{\rm emp}$ & ESS & Steps & Proj./walk & RMSE $\mathfrak m^*=0$ & RMSE Hui \\
\midrule
3&0&.667&.664&1.000&12.75&57.73&.0408&.0406\\
3&.5&.672&.673&.996&12.51&53.10&.0393&.0389\\
3&1&.687&.688&.940&12.70&57.74&.0476&.0469\\
3&2&.731&.733&.614&11.57&50.82&.0752&.0752\\
3&4&.812&.813&.205&10.46&43.74&.1315&.1314\\
4&0&.500&.502&1.000&16.31&55.85&.0557&.0558\\
4&.5&.505&.506&.995&16.75&59.18&.0468&.0468\\
4&1&.520&.526&.945&15.74&54.03&.0534&.0534\\
4&2&.567&.572&.566&14.05&50.05&.0840&.0839\\
4&4&.671&.666&.152&12.60&43.91&.6013&.6014\\
6&0&.333&.336&1.000&27.48&76.68&.0142&.0142\\
6&.5&.337&.335&.994&26.98&75.60&.0128&.0131\\
6&1&.347&.348&.912&27.09&75.13&.0203&.0201\\
6&2&.382&.382&.414&22.81&65.61&.0975&.0973\\
6&4&.480&.478&.120&18.86&53.91&.2487&.2487\\
8&0&.250&.250&1.000&37.68&76.72&.0095&.0094\\
8&.5&.252&.253&.995&34.51&69.87&.0075&.0075\\
8&1&.259&.260&.918&35.78&73.20&.0158&.0164\\
8&2&.285&.281&.472&34.26&69.78&.1457&.1464\\
8&4&.363&.366&.121&27.62&55.50&.3120&.3120\\
\bottomrule
\end{tabular}
\end{table*}

The complete sweep supports the same conclusion as the selected main-text
rows: increasing $\kappa$ improves local tangent progress but eventually
collapses path ESS, while increasing $D$ trades projection separation against
the $m/D$ tangent penalty. The relevant design criterion is therefore cost at
a requested accuracy rather than walk length alone.

\subsection{Numerical validation of the curvature-corrected step law}
The curvature-corrected projected-step law predicts that the leading
$m/D$ tangent-progress term is modified at order $r^4$ by the intrinsic and
extrinsic curvature of the embedding. We validate this correction independently
of the Swiss-roll solver using geometries with analytic closest-point maps:
a unit sphere in $\mathbb R^3$, the same sphere zero-padded into
$\mathbb R^6$, and a unit cylinder in $\mathbb R^3$.

For each geometry, we use $r=0.1$ and four million antithetic direction pairs.
An exact flat-term control variate removes the dominant $r^2$ contribution,
allowing the $r^4$ coefficient to be estimated directly. Table~\ref{tab:curvature_step_validation}
compares the resulting scalar coefficient and covariance correction with the
theoretical prediction.

\begin{table}[t]
\centering
\caption{Numerical validation of the $r^4$ curvature correction in the
projected-step law. The scalar coefficient and full covariance correction agree
with the theoretical prediction at the percent level or better.}
\label{tab:curvature_step_validation}
\small
\begin{tabular}{lrrrr}
\toprule
Surface & $D$ & Theory $c_4$ & Empirical $c_4$ & Tensor rel. err. \\
\midrule
Sphere   & 3 & .044444 & .044857 & .0113 \\
Sphere   & 6 & .013889 & .013827 & .0052 \\
Cylinder & 3 & .066667 & .066943 & .0042 \\
\bottomrule
\end{tabular}
\end{table}

The agreement supports two parts of the analysis separately from the full PDE
experiment. First, the flat $m/D$ law remains the leading contribution after
projection. Second, finite-radius progress depends on the particular
isometric embedding through the curvature correction, so two embeddings of
the same intrinsic surface can have different projected-walk efficiency even
before global projection conflicts are considered.

% \paragraph{Reproducibility.}The experiment is implemented in\texttt{\detokenize{code/curvature_projected_step_validation.py}} and thereported values are stored in \texttt{\detokenize{experiment_data/curvature_projected_step_validation.csv}}.

\subsection{Qualitative validation on an accordion ribbon}
\label{sec:accordion}

The primary experiments use a Swiss-roll geometry, whose competing branches
are generated mainly by radial winding. To check that the projection-state
construction is not tied to this particular geometry, we additionally apply
the same embedding, phase-field ownership, and fast/transition preprocessing
to an accordion ribbon with repeated translational folds.

The owner field separates nearby folds into coherent projection states, while
transition regions concentrate around locations where continuation can move
between owners. This reproduces the intended mechanism on a folding pattern
that is geometrically distinct from the Swiss roll. 

% We use this example only as a qualitative representation check; no quantitative accuracy, runtime, or codimension claim in the main text depends on it.

\begin{figure}[t]
    \centering
    \includegraphics[width=0.9\linewidth]{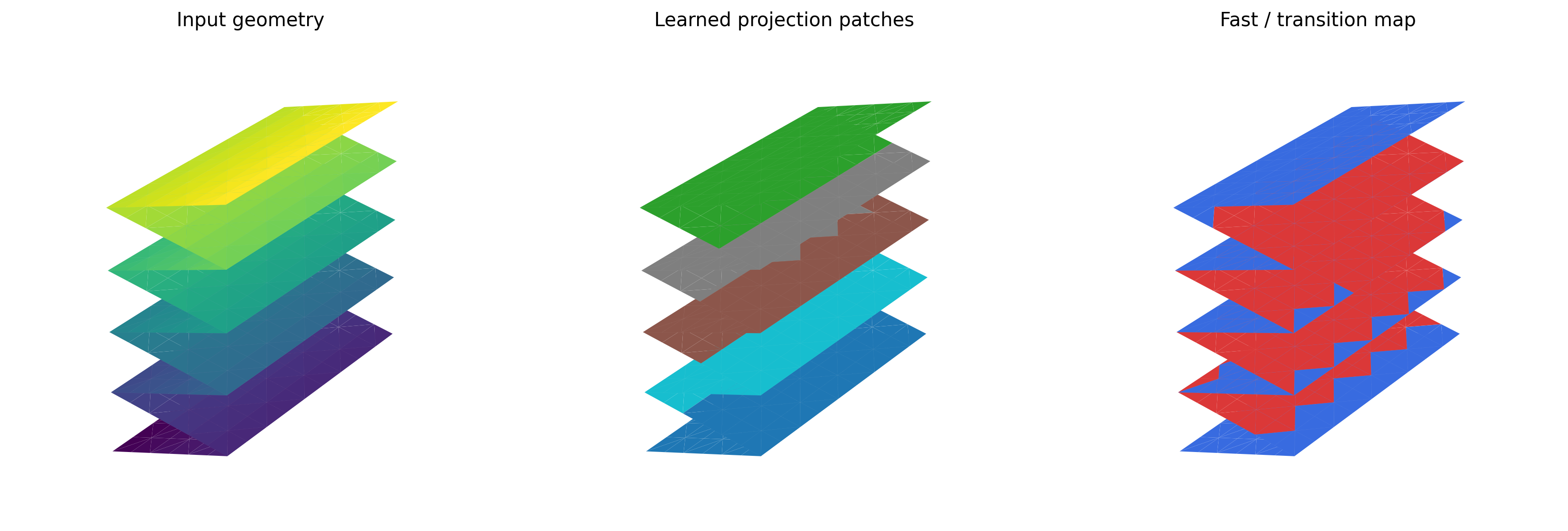}
    \caption{Accordion-ribbon qualitative check. From left to right: the
    self-approaching input geometry, projection-owner decomposition, and the
    empirical fast/transition map obtained by the same preprocessing used for
    the Swiss-roll experiments.}
    \label{fig:exp-accordion-geometry}
    \Description{Accordion-ribbon geometry, projection owners, and empirical
    fast/transition regions.}
\end{figure}

%  We use this example only
% as a qualitative representation check; no quantitative accuracy, runtime, or
% codimension claim in the main text depends on it.

\begin{figure}[t]
    \centering
    \includegraphics[width=0.9\linewidth]{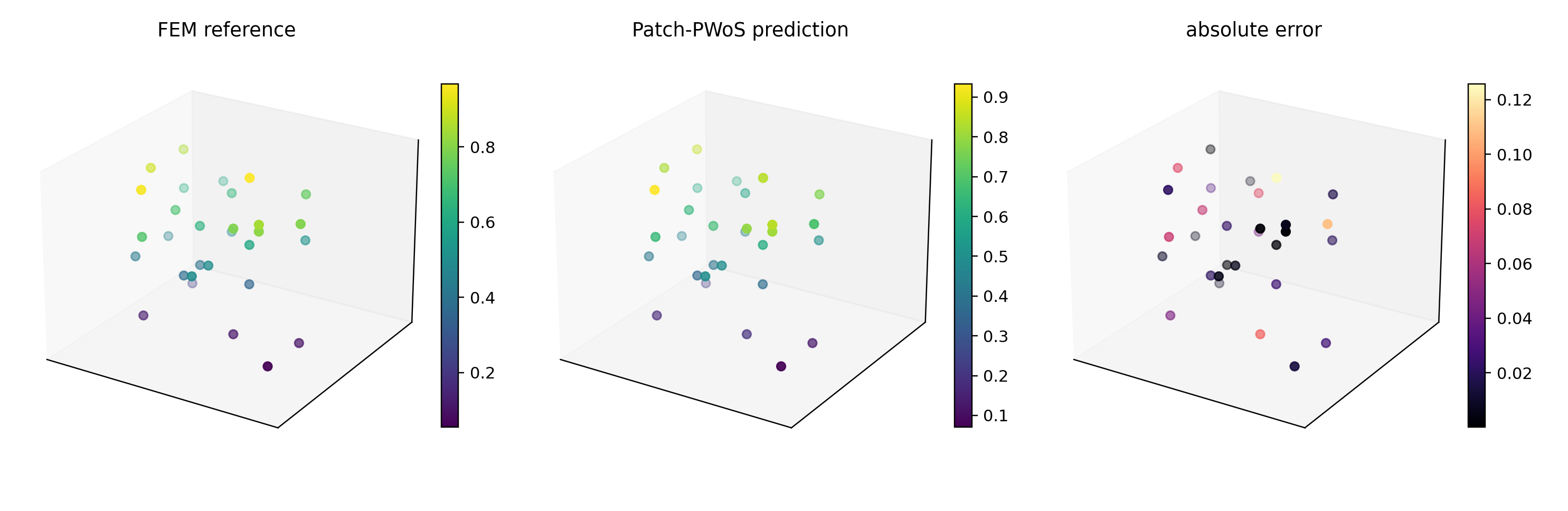}
    \caption{Qualitative PDE result on the accordion ribbon: reference solution,
    WoAP estimate, and absolute error at the same query construction.}
    \label{fig:exp-accordion-pde}
    \Description{Reference solution, WoAP estimate, and absolute error on the
    accordion ribbon.}
\end{figure}

\begin{figure}[t]
    \centering
    \includegraphics[width=0.9\linewidth]{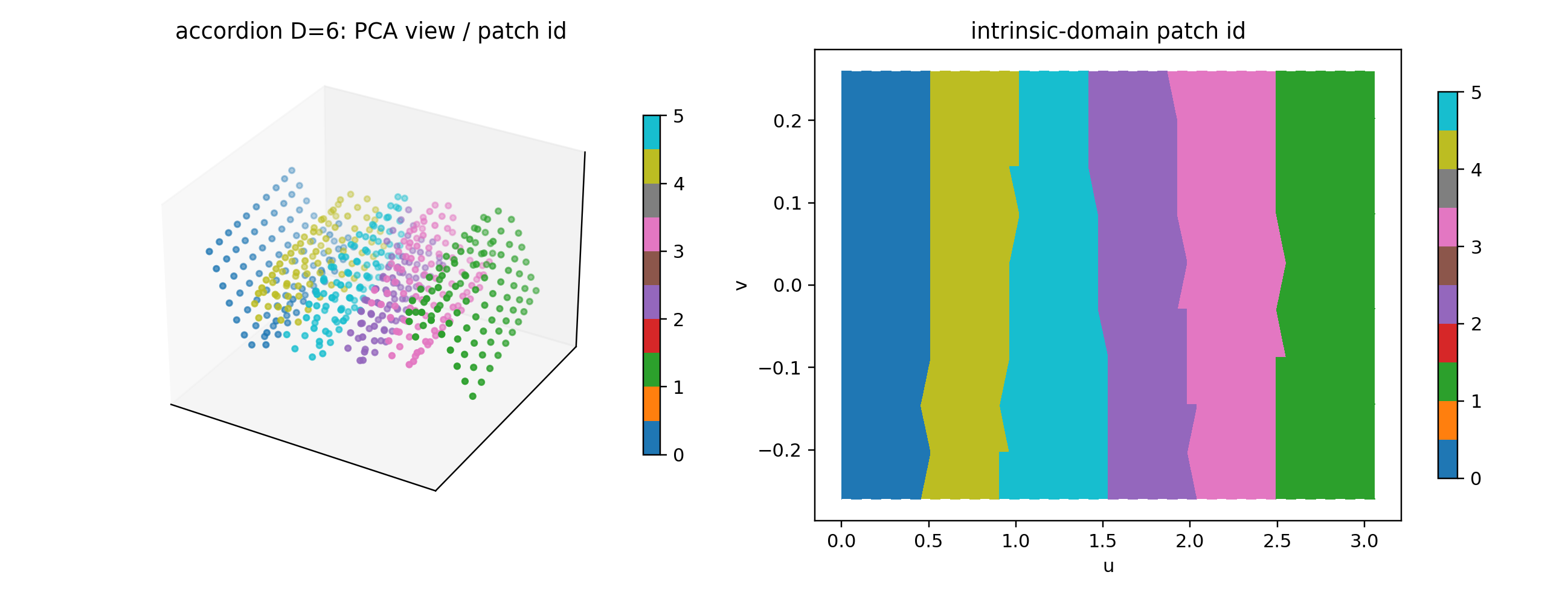}
    \caption{Lightweight $D=6$ re-embedding of the accordion ribbon, shown by
    a three-dimensional PCA view together with its intrinsic-domain projection
    states. This visualization is not used in the quantitative dimension
    sweep.}
    \label{fig:exp-accordion-highdim}
    \Description{High-dimensional accordion-ribbon re-embedding and projection
    states.}
\end{figure}
\end{document}